\documentclass[reprint,twocolumn,superscriptaddress,nofootinbib,floatfix]{revtex4-2}
 
\usepackage{float}
\usepackage{verbatim}
\usepackage[utf8]{inputenc}
\usepackage[T1]{fontenc}
\usepackage[english]{babel}
\usepackage{dsfont}
\usepackage[normalem]{ulem}
\usepackage{amsmath}
\usepackage{amssymb}
\usepackage{mathrsfs}
\usepackage{amsthm}
\usepackage{mathtools}
\usepackage{wasysym}
\usepackage{comment}
\usepackage{bbm}
\usepackage{bm}
\usepackage{nicefrac}
\usepackage{soul}
\usepackage{stackrel}
\usepackage{braket}
\usepackage{enumitem}
\usepackage{graphicx}
\usepackage{tikz}
\usepackage{mdframed}
\usepackage{booktabs}
\usepackage{siunitx}

\usepackage{calc}
\usepackage{accents}

\usepackage{multirow}
\usepackage{array}[=2016-10-06] 
\newcolumntype{P}[1]{>{\centering\arraybackslash}p{#1}}

\usepackage[hidelinks]{hyperref}
\usepackage{framed}

\addto\captionsenglish{}

\newcommand{\abs}[1]{\left|#1\right|}

\DeclarePairedDelimiter{\ceil}{\lceil}{\rceil}
\DeclarePairedDelimiter\floor{\lfloor}{\rfloor}
\usepackage{cases}

\newtheorem{theorem}{Theorem}[section]

\newtheorem{lemma}{Lemma}[section]
\newtheorem*{lemma*}{Lemma}
\newtheorem{definition}{Definition}[section]

\addto\captionsenglish{}
\addto\captionsenglish{}

\usepackage{algorithm}

\makeatletter
\newenvironment{protocol}
{
		\renewcommand{\ALG@name}{Protocol}
		\refstepcounter{algorithm}
		\hrule height.8pt depth0pt \kern2pt
		\renewcommand{\caption}[2][\relax]{
			{\raggedright\textbf{\fname@algorithm~\thealgorithm} ##2\par}%
			\ifx\relax##1\relax 
			\addcontentsline{loa}{algorithm}{\protect\numberline{\thealgorithm}##2}%
			\else 
			\addcontentsline{loa}{algorithm}{\protect\numberline{\thealgorithm}##1}%
			\fi
			\kern2pt\hrule\kern2pt
		}
	}{
	\kern2pt\hrule\relax
}
\makeatother

\usepackage{natbib}
\makeatletter
\let\cat@comma@active\@empty
\makeatother

\begin{document}
\title{Equivocation-resistant multiparty digital signature for quantum networks}

\author{Federico Grasselli}
\author{Gaetano Russo}
\author{Giuseppe De Falco}
\affiliation{Leonardo Innovation Labs -- Quantum Technologies, Leonardo S.p.A., Via Tiburtina km 12400, 00131 Rome, Italy}
\author{Stefano Pepe}
\affiliation{Optronics, Laser \& Quantum Technologies, Leonardo S.p.A. Electronics Division, Via Tiburtina km 12400, 00131 Rome, Italy}
\author{Carlo Liorni}
\author{Massimiliano Proietti}
\affiliation{Leonardo Innovation Labs -- Quantum Technologies, Leonardo S.p.A., Via Tiburtina km 12400, 00131 Rome, Italy}

\begin{abstract}
    Digital signatures are a critical cryptographic primitive requiring quantum-safe solutions. One possibility are quantum digital signatures (QDS), which offer information-theoretic (IT) security without a trusted authority. However, even the most promising QDS proposals are largely limited to the tripartite scenario (one sender, two receivers) and are vulnerable to equivocation-based attacks that hinder transferability and non-repudiation. To overcome such limitations, we introduce an equivocation-resistant signature (ERS) protocol based on preshared keys and universal hashing that achieves IT security and scales to an arbitrary number of receivers. We benchmark the ERS protocol against state-of-the-art QDS schemes---namely [Amiri \textit{et al.}, 2018] for the multi-receiver scenario and [Yin \textit{et al.}, 2023] and [García Cid \textit{et al.}, 2025] for the tripartite case---demonstrating orders-of-magnitude reductions in preshared key consumption and signature size. The superior performance of the ERS protocol is also validated in a field test on the Rome quantum metropolitan area network, reaching a rate of 13 signatures per second for one-megabit documents shared among five parties. Our findings position our ERS protocol as a strong candidate for implementing IT-secure digital signatures in today's quantum communication infrastructures.
\end{abstract}

\maketitle

\section{Introduction}
Popular public-key digital signature (DS) protocols, such as the digital signature algorithm (DSA) \cite{DSA} and the elliptic curves DSA \cite{ECDSA}, can be broken by large-scale fault-tolerant quantum computers. The new post-quantum DS standards FIPS 204 \cite{FIPS204} and FIPS 205 \cite{FIPS205} are not (yet) efficiently attackable by quantum computers, but their security is still based on computational assumptions. Conversely, the appeal of quantum digital signatures (QDS) lies in their information-theoretic (IT) security, i.e., security against computationally unbounded adversaries, analogously to quantum key distribution (QKD).

DSs based on a public-key infrastructure can be verified by any user by simply retrieving the public key of the sender, which requires a trusted third party (a certificate authority) certifying that the public key actually belongs to the sender. Conversely, QDS schemes aiming at IT security mostly avoid a trusted authority and may require pairwise preshared keys between all users, which can limit the number of users able to verify a signature. For this reason, QDS schemes have been mostly analyzed in the simplest non-trivial scenario, i.e. a tripartite scenario with a sender and two receivers (or verifiers).

The first QDS was proposed in 2001 by Gottesman and Chuang \cite{QDS-GottesmanChuang}, with challenging experimental requirements such as quantum memories, secure\footnote{Note that an authenticated quantum channel must be also a private channel. Indeed, any attempt to learn about the quantum state traveling in the channel leads to a modification of the state, which implies the channel is no longer authentic. Hence, authenticity implies privacy for quantum channels.} quantum channels and swap tests. Later protocols removed the requirement of quantum memories \cite{QDS-noqmemory} and the need of secure quantum channels \cite{QDS-2016,Amiri2016}. However, the main drawback of such QDS schemes remained, i.e., the need to sign single-bit messages, thus resulting inefficient for long messages.

To overcome the security and efficiency limitations of previous QDS proposals, three promising QDS schemes have been proposed, respectively by Yin \textit{et al.} in Ref.~\cite{QDS-Chen}, by García Cid \textit{et al.} in Ref.~\cite{QDS-INDRA} and by Amiri \textit{et al.} in Ref.~\cite{USS-andersson}. These schemes can sign arbitrarily long documents with relatively short signatures, thereby improving on the efficiency of previous proposals by several orders of magnitude. The quantum nature of these schemes lies in leveraging preshared QKD keys  in otherwise classical algorithms, hence they are also known as quantum-assisted or unconditionally-secure digital signatures.

Recent improvements on the above schemes include simplifications \cite{QDS-Chen2,QDS-Chen3} on the QKD protocols preceding the classical signature algorithm of Ref.~\cite{QDS-Chen}, performance and security upgrades \cite{QDS-CNR} on the scheme from Ref.~\cite{QDS-INDRA}, and a modified version \cite{QDS-Kiktenko} of the scheme from Ref.~\cite{USS-andersson} that closes a security loophole of the original protocol and improves the scalability for large networks of receivers -- at the cost of limiting the number of dishonest users.

Despite the progress on practicality and security, the protocols in Refs.~\cite{QDS-Chen,QDS-INDRA} and their improved versions are still fundamentally limited to the tripartite scenario, as reflected by their experimental implementations \cite{QDS-Chen,QDS-INDRA,QDS-CNR,QDS-ScienceExp}. The only experimental realization beyond the tripartite scenario \cite{QDS-exp8network} implements the QDS protocol of Ref.~\cite{USS-andersson}, which is inherently designed for an arbitrary number of receivers.

\subsection{Contributions of this work}

In this work, we start by systematically reviewing the security of the three promising QDS protocols of Refs.~\cite{QDS-Chen,QDS-INDRA,USS-andersson} focusing on to the need for authenticated (and reliable) channels to avoid security loopholes. We uniform the security claims of the three protocols by proving their IT security while taking into account the failure probability of IT-secure authenticated channels, unlike previous works \cite{QDS-Chen,QDS-INDRA,USS-andersson,QDS-Kiktenko,QDS-CNR,QDS-Chen3,QDS-Chen2}. In the case of the protocol by García Cid \textit{et al.} \cite{QDS-INDRA}, we propose a modified protocol which achieves IT security by replacing computationally secure hashes with $\varepsilon$-almost XOR universal$_2$ hash families \cite{LFSR-hashing}. For the protocol by Amiri \textit{et al.} \cite{USS-andersson}, we address the security loophole raised in Ref.~\cite{QDS-Kiktenko} with a complementary approach to that of Ref.~\cite{QDS-Kiktenko}.

From our security analyses, we identify a critical flaw caused by equivocation attacks in the protocol of Ref.~\cite{USS-andersson}. Indeed, dishonest receivers can effectively prevent the transferability of a document-signature pair by modifying it during the transfer to honest receivers without being noticed. By the same mechanism, a malicious sender colluding with dishonest receivers can repudiate an authentic signature by ensuring that enough honest receivers evaluate distinct document-signature pairs. We remark that this flaw is similarly shared by several other multiparty QDS proposals, such as Refs.~\cite{multiQDS-Chen,multiQDS-Korolkova,potentialmultiQDS-Toshiba,multiQDS-Ma}, which however are not considered in this work due to their limitation of signing one-bit messages.

While we recognize that IT-secure signature schemes are usually not expected to ensure consensus on the document-signature pair in multiparty scenarios, we consider this an important requirement. Indeed, without global consensus on the actual signature, transferability failures and repudiations could be trivially enforced, circumventing the security claims of the respective QDS schemes.

To overcome this critical vulnerability affecting QDS protocols in multiparty scenarios \cite{USS-andersson} and enhance scalability of previous QDS schemes \cite{QDS-Chen,QDS-INDRA}, we introduce an equivocation-resistant signature (ERS) protocol that natively supports an arbitrary number of receivers and prevents equivocation-based attacks, without requiring a trusted authority. The ERS protocol has been filed for a patent.

We benchmark the performance of our ERS protocol against the three promising QDS protocols of Refs.~\cite{QDS-Chen,QDS-INDRA,USS-andersson} by numerically optimizing their parameters when signing documents of a given size at a given security threshold. In the multiparty scenario, our ERS protocol reduces the consumed preshared bits by two orders of magnitude and the signature size by up to five orders of magnitude, compared to the protocol of Ref.~\cite{USS-andersson}. In the tripartite scenario, where all three QDS schemes of Refs.~\cite{QDS-Chen,QDS-INDRA,USS-andersson} apply, the ERS protocol matches the performance of the best-performing QDS protocol.

By grounding its IT security in an efficient use of preshared QKD keys, our ERS protocol could be readily implemented in today's point-to-point QKD networks. We demonstrate this by performing a field test on the Rome quantum metropolitan area network.

\subsection{Paper structure}

The paper is organized as follows. In Sec.~\ref{sec:notation} we describe the adversarial scenario and set the notation for the rest of the paper. In Sec.~\ref{sec:all-QDS} we provide a security analysis of three selected QDS protocols. In Sec.~\ref{sec:ERS} we describe the ERS protocol and discuss its security. We benchmark the performance of the ERS protocol in the tripartite and multipartite scenarios in Sec.~\ref{sec:performance} and validate its performance in a field test in Sec.~\ref{sec:experiment}. We draw conclusions in Sec.~\ref{sec:conclusion}.

Regarding the appendices, Appendix~\ref{app:WC} provides background on Wegman-Carter message authentication while Appendix~\ref{app:efficient-hashing} details the universal hashing families adopted in the manuscript.  Appendices~\ref{app:security-Chen}, \ref{app:security-INDRA}, and \ref{app:security-Andersson} contain the security proofs of the QDS protocols from Ref.~\cite{QDS-Chen}, Ref.~\cite{QDS-INDRA}, and Ref.~\cite{USS-andersson}, respectively. In Appendix~\ref{app:BRB} we detail Bracha's reliable broadcast (used in the ERS protocol). We provide a full security proof for our ERS protocol in Appendix~\ref{app:security-ERS}.

\section{Notation} \label{sec:notation}

In order to systematically analyze the security of the three QDS protocols from Refs.~\cite{QDS-Chen,QDS-INDRA,USS-andersson} in the tripartite scenario, we adopt a uniform notation.

The document to be signed is indicated as $Doc$ and the corresponding signature is $Sig$. Together, they form the document-signature pair $\{Doc,Sig\}$. The bit length of $Doc$ is $b_M$.

In our tripartite scenario, Alice is the sender and produces the pair $\{Doc,Sig\}$, Bob is the receiver of the pair $\{Doc,Sig\}$ and Charlie is the honest verifier. Alice and Bob can be malicious, but not both at the same time. Specifically, a malicious Alice can perform a repudiation attack, in which she wants to convince Charlie that the pair $\{Doc,Sig\}$ verified by Bob is not authentic. Alternatively, a malicious Bob can perform a forgery attack where he produces a forged pair $\{Doc',Sig'\}$ and wants Charlie to accept the forged pair. We prove the security of the analyzed QDS protocols against forgery and repudiation attacks, according to the following definitions.

\begin{definition} \label{def:forgery}
    A QDS protocol is $\varepsilon_{\rm for}$-secure against forgery attacks if the probability that the protocol does not abort and that Charlie accepts a pair $\{Doc',Sig'\}$ forged by Bob is at most $\varepsilon_{\rm for}$.
\end{definition}

\begin{definition} \label{def:repudiation}
    A QDS protocol is $\varepsilon_{\rm rep}$-secure against repudiation attacks if the probability that the protocol does not abort and that Charlie rejects a document-signature pair accepted by Bob is at most $\varepsilon_{\rm rep}$.
\end{definition}

The analyzed QDS protocols \cite{QDS-Chen,QDS-INDRA,USS-andersson} achieve IT security with respect to forgery and repudiation attacks by extending IT-secure message authentication codes --specifically, Wegman-Carter (WC) authentication based on universal hashing \cite{WC}-- to scenarios with multiple receivers, while making sure that no receiver can pretend to be the sender. In the considered QDS protocols, Alice generates the signature by hashing the $b_M$-bit document with hash functions derived from either an $\varepsilon$-almost XOR universal$_2$ family or an $\varepsilon$-almost strongly universal$_2$ family, generating hashes of $b_H$ bits each. We define both families for completeness below.

\begin{definition} \label{def:almost-strongly-universal2}
Let $\mathcal{F}_{\rm ASU}=\{f:M \to B\}$ be a family of hash functions. Then, $\mathcal{F}_{\rm ASU}$ is $\varepsilon$-almost strongly universal$_2$ ($\varepsilon$-ASU$_2$) if the following two conditions hold:
\begin{align}
    &\bullet \Pr_{f \in_R \mathcal{F}_{\rm ASU}} [f(m) = b ] = \frac{1}{2^{b_H}} \quad \forall\, m \in M  \,\, \forall\, b \in B ,  \label{eq-almost-strongly-universal1} \\
    &\bullet \Pr_{f \in_R \mathcal{F}_{\rm ASU}} [f(m_1) = b_1 \wedge f(m_2)=b_2] \leq \frac{\varepsilon}{2^{b_H}}  \nonumber\\
    &\quad\quad\quad\forall\, m_1 \neq m_2, \, \forall\, b_1,b_2 \in B,  \label{eq-almost-strongly-universal2}
\end{align}
where $f\in_R \mathcal{F}$ indicates that the function $f$ is sampled randomly from the set $\mathcal{F}$ and where $b_H$ is the bit length of the elements in $B$.
\end{definition}

\begin{definition} \label{def:almost-XOR-universal2}
Let $\mathcal{F}_{\rm AXU}=\{f:M \to B\}$ be a family of hash functions. Then, $\mathcal{F}_{\rm AXU}$ is $\varepsilon$-almost XOR universal$_2$ ($\varepsilon$-AXU$_2$) if:
\begin{align}
        \Pr_{f \in_R \mathcal{F}_{\rm AXU}} [f(m_1) \oplus f(m_2) = b] \leq \varepsilon \quad \forall\, m_1 \neq m_2,\,\forall\, b. \label{eq-almost-XOR-universal2}
\end{align}
\end{definition}
Note that $\varepsilon$-ASU$_2$ families are a subset of $\varepsilon$-AXU$_2$ families since \eqref{eq-almost-strongly-universal2} implies \eqref{eq-almost-XOR-universal2}.

In Table~\ref{tab:efficient-hashing} we report the two specific families adopted by the analyzed QDS protocols, summarizing their main characteristics. The exact definition of the two families is reported in Appendix~\ref{app:efficient-hashing}.

\begin{table}[h!t]
\renewcommand\arraystretch{1.3}
\setlength{\tabcolsep}{1pt}
\centering
\caption{The two universal hashing families considered in this manuscript, mapping messages of $b_M$ bits to tags of $b_H$ bits. We report the number of preshared key bits consumed to agree on a specific function of the family and the security parameter. More details are found in Appendix~\ref{app:efficient-hashing}.}
\label{tab:efficient-hashing}
\begin{tabular}[t]{|p{0.25\linewidth}|p{0.45\linewidth}|p{0.15\linewidth}|}
\toprule
\textbf{Family} & \textbf{Preshared bits} & $\varepsilon$ \\
\midrule
 $\mathcal{F}_{\rm ASU}$ \cite{ASU-Andersson,QDS-Kiktenko} & $3 b_H + 2  \log_2 \left(\frac{b_M}{b_H} -1\right)$ & $2^{1-b_H}$\\
$\mathcal{F}_{\rm AXU}$ \cite{LFSR-hashing} & $2b_H$ & $b_M 2^{1-b_H}$\\
\bottomrule
\end{tabular}
\end{table}

All authenticated channels employed in the protocols of this work are assumed to be reliable (no message is lost) and are implemented with IT security via WC authentication (see Appendix~\ref{app:WC} for more details). In each QDS protocol, it is implicitly assumed that failing to authenticate a message causes the protocol to abort (which is not equivalent to rejecting the signature). Moreover, in our security analyses we account for the rare events in which a forged message is successfully authenticated. That is, we account for the failure probability of WC authentication. In order to compare the QDS protocols on an equal footing, we assume that all authenticated channels implement WC authentication with key recycling with hash functions from the $\mathcal{F}_{\rm AXU}$ family (see Table~\ref{tab:efficient-hashing}), producing hashes of $b'_H$ bits that are protected by OTP with fresh secret bits. Therefore, the number of preshared bits consumed by sending $n$ messages over an authenticated channel is $(2+n)b'_H$, where we also accounted for the bits required to agree on a specific element from $\mathcal{F}_{\rm AXU}$.

\section{Security analysis of three QDS protocols} \label{sec:all-QDS}

In this section we analyze the security of three practical QDS protocols that are based on the protocols introduced in Refs.~\cite{QDS-Chen,QDS-INDRA,USS-andersson}. For each protocol, we provide a detailed description highlighting the differences with respect to its original formulation and a security proof against forgery and repudiation attacks as defined in Sec.~\ref{sec:notation}. Our security proofs account for the failure probability of WC authentication, unlike their original derivations, while all preshared secret bits are assumed to be perfectly secure. In the last subsection, we compare the security approaches of the three QDS protocols. 

\subsection{QDS by Yin \textit{et al.}} \label{sec:QDS-Chen}

The tripartite QDS scheme introduced in Ref.~\cite{QDS-Chen} resembles a tripartite classical DS scheme, where the security is lifted to IT security by replacing public-key cryptography with the use of one-time pad (OTP) and one-time universal$_2$ hashing (OTUH), i.e., hash functions from the $\mathcal{F}_{\rm AXU}$ family that are renewed for each document signature.

\subsubsection{The protocol}
We modify the original formulation of the protocol from \cite{QDS-Chen} such that Charlie always verifies the signature after receiving the outcome of the verification of Bob, regardless of Bob accepting it or not. In this way, Charlie can detect if Bob is a liar who accepts falsified signatures or rejects original signatures. We believe this is an additional useful feature of the protocol.\\

\begin{protocol}  \label{QDSprot-Chen} \caption{QDS protocol \cite{QDS-Chen}}
\begin{enumerate}[wide, labelwidth=!, labelindent=0pt]

\item \textit{Distribution stage}\quad Alice, Bob and Charlie run a QKD or quantum secret sharing protocol, in order to establish secret keys $X_A$ for Alice, $X_B$ for Bob, and $X_C$ for Charlie, with the property that $X_A=X_B\oplus X_C$. We partition each of the keys into two partitions. The first $b_H$ bits of each key is labeled as $X^{b_H}_A$, $X^{b_H}_B$, and $X^{b_H}_C$, respectively, while the following $2{b_H}$ bits are labeled as $X^{2b_H}_A$, $X^{2b_H}_B$, and $X^{2b_H}_C$. It holds: $X_A = X^{b_H}_A || X^{2b_H}_A$ and similarly for $X_B$ and $X_C$.

\item  \textit{Signing of Alice}\quad Alice generates a $b_H$-bit random string, $p_a$, with elements $(p_a)_i$. After associating to $p_a$ the following polynomial over GF(2): $p_a(x)=x^{b_H} + (p_a)_{b_H-1} x^{b_H-1} + \dots + (p_a)_1 x + (p_a)_0$ of degree $b_H$, Alice checks whether $p_a(x)$ is irreducible (see algorithm in Supplementary Material of \cite{QDS-Chen}). If the test is negative, Alice generates a new random string until the corresponding polynomial is irreducible. From the irreducible polynomial $p_a(x)$ and the key $X_A^{b_H}$, Alice defines a linear feedback shift register (LFSR) and obtains the associated Toeplitz matrix $T_{p_a,X_A^{b_H}}$, which is an element of $\mathcal{F}_{\rm AXU}$ (see Appendix~\ref{app:efficient-hashing}). Then, she computes the $b_H$-bit hash value of the document: $h_a = T_{p_a,X_A^{b_H}} \cdot Doc$ and the digest: $Dig = (h_a||p_a)$ as the concatenation of the hash of the document with the random string $p_a$. Finally, Alice derives the signature by encrypting the digest with the secret key $X^{2b_H}_A$ via OTP: $Sig = Dig \oplus X^{2b_H}_A$. The couple $\{Doc,Sig\}$ is sent to the receiver, Bob, over a public channel.

\item  \textit{Verification of Bob}\quad Firstly, Bob sends via an authenticated channel the received couple $\{Doc,Sig\}$ and the key $X_B$ to Charlie. Once Charlie receives the data from Bob, he sends $X_C$ to Bob over the same authenticated channel. Bob uses the key from Charlie to recover Alice's key, by computing: $K_B = X_B \oplus X_C$. Using $K^{2b_H}_B$ Bob recovers the digest by computing $K^{2b_H}_B \oplus Sig$ and retrieves, in particular, the bit string $p_b$ and the hash $h_b$ produced by Alice. Bob generates the LFSR-based Toepliz matrix using $p_b$ and the string $K^{b_H}_B$, $T_{p_b,K^{b_H}_B}$, and computes the hash $h_b'=T_{p_b,K^{b_H}_B} \cdot Doc$. Bob accepts the signature if $h_b=h_b'$ and informs Charlie by sending him the bit $V_B=0$ over the authenticated channel. Otherwise, if $h_b\neq h_b'$, Bob rejects the signature and sends $V_B=1$ to Charlie.

\item  \textit{Verification of Charlie}\quad Charlie uses his key and the key received from Bob to compute $K_C = X_C \oplus X_B$. Then, he employs $K^{2b_H}_C$ and the $Sig$ from Bob to acquire an expected digest: $K^{2b_H}_C \oplus Sig$, i.e. the bit string $p_c$ and the hash $h_c$ produced by Alice. Subsequently, Charlie generates the LFSR-based Toeplitz matrix $T_{p_c,K^{b_H}_C}$ and computes the hash $h_c'=T_{p_c,K^{b_H}_C} \cdot Doc$. Charlie accepts the signature if $h_c=h_c'$, otherwise he rejects it.

\end{enumerate}

\end{protocol}

\subsubsection{Security proof}

In the following, we rederive the security of Protocol~\ref{QDSprot-Chen} with respect to forgery and repudiation. In doing so, we account for the  failure probability of the authenticated channel between Bob and Charlie, which was not considered in the original paper \cite{QDS-Chen}. Moreover, we show that failing to renew the hash function from $\mathcal{F}_{\rm AXU}$ for each document signature opens a security loophole. The proofs are found in Appendix~\ref{app:security-Chen}.

\begin{lemma} \label{lm:forgery-Chen}
    Protocol~\ref{QDSprot-Chen} is $\varepsilon_{\rm for}$-secure against forgery according to Definition~\ref{def:forgery}, with $\varepsilon_{\rm for}=b_M/2^{b_H-1}$.
\end{lemma}

\begin{lemma} \label{lm:repudiation-Chen}
    Protocol~\ref{QDSprot-Chen}, with $b_M>b_H$, is $\varepsilon_{\rm rep}$-secure against repudiation according to Definition~\ref{def:repudiation}, with $\varepsilon_{\rm rep}=(2b_H+b_M)/2^{b'_H-1}$.
\end{lemma}

Although forgery attacks in Protocol~1 are related to attacks in WC authentication, there is an important difference. While in WC authentication a given hash function can be reused for subsequent messages (key recycling) when the tags are protected by OTP, in Protocol~1 this is not possible since the hash function becomes known to Bob after the protocol execution and Bob is a potential adversary. This implies that Alice must employ a new function from  $\mathcal{F}_{\rm AXU}$ for every new signature, i.e., both the initial vector $X^{b_H}_A$ and the irreducible polynomial $p_a$ must be renewed.

In the following we prove that renewing $X^{b_H}_A$ without renewing $p_a$ opens a security loophole. Indeed suppose Alice chooses the same value for $p_a$ to sign multiple documents. Then, Bob will know $p_a$ when he receives the second document signed by Alice. This enables Bob to forge a $\{Doc',Sig'\}$ pair that will be validated by Charlie, thereby making the protocol vulnerable to forgery attacks.

\begin{lemma} \label{lm:loophole-Chen}
    If $p_a$ is known to Bob, then he can always perform a successful forgery attack, i.e., Charlie validates the pair $\{Doc',Sig'\}$ forged by Bob with unit probability.
\end{lemma}

\subsection{QDS by García Cid \textit{et al.}} \label{sec:QDS-Indra}

In Ref.~\cite{QDS-INDRA}, García Cid \textit{et al.} propose a tripartite QDS scheme that combines QKD symmetric keys with NIST-recommended hash functions, which, however, are not IT secure, thus preventing the whole QDS scheme from being so. To make a fair comparison with the other QDS protocols addressed in this manuscript, we describe a modified version of the protocol from \cite{QDS-INDRA}, where we replace the computationally-secure hash functions with the $\mathcal{F}_{\rm AXU}$ family of hash functions described in Sec.~\ref{sec:notation}. Moreover, we require that Bob forwards the document-signature pair to Charlie via an authenticated channel, otherwise Alice could easily enforce a repudiation.

We point out that, in a recent work \cite{GeneralizedQDS}, the authors modified the original protocol from \cite{QDS-INDRA} in order to increase its security and efficiency. Although the scheme from Ref.~\cite{GeneralizedQDS} closely resembles the version we propose below, it is still not IT-secure as successful forgery attacks are always possible by exhaustive search, while this is not possible in our modified version.

\subsubsection{The modified protocol}

Similarly to the protocol by Yin \textit{et al.}, Alice establishes symmetric keys with Bob and Charlie with QKD. However, rather than combining the two keys via the XOR operation, she concatenates the two keys into a single key. Similarly to the protocol by Amiri \textit{et al.}, the protocol requires Bob and Charlie to communicate via authenticated secret channels and exchange a random subset of their symmetric keys. In this way, Alice cannot easily force a repudiation attack by adding noise in the signature destined to one specific receiver, since she does not know anymore the key held by each receiver.\\

\begin{protocol}  \label{QDSprot-INDRA} \caption{QDS protocol (modified from \cite{QDS-INDRA})}
\begin{enumerate}[wide, labelwidth=!, labelindent=0pt]

\item \textit{Distribution stage}\quad Alice establishes symmetric keys with Bob and Charlie. Bob and Charlie exchange random key blocks.

\begin{enumerate}[label*=\arabic*.]
    \item Alice runs a QKD protocol with Bob (Charlie) and establishes a shared secret key $X_B$ ($X_C$) of length $3n b_H$ bits, which can be divided into $n$ blocks of length $3b_H$ each: $X_B=X^1_B || X^2_B || \dots || X^n_B $ for Bob's key and as $X_C= X_C^1 || X_C^2 || \dots || X_C^n$ for Charlie's key, where the ''$||$'' symbol indicates concatenation. Each block in Bob's key (Charlie's key) can be decomposed as: $X_B^j=s^j_B || r^j_B$ ($X_C^j=s^j_C || r^j_C$), where $s_B^j$ ($s_C^j$) is $b_H$ bits long and $r^j_B$ ($r^j_C$) is $2b_H$ bits long.

    \item Bob (Charlie) selects a random permutation $\gamma_B \in S_n$ ($\gamma_C \in S_n$) from the set $S_n$ of permutations of $n$ elements and applies the permutation on their key blocks. Bob obtains the reshuffled key: $X_B':= X_B^{\gamma_B(1)} || X_B^{\gamma_B(2)}|| \dots || X_B^{\gamma_B(n)}$ and Charlie obtains the reshuffled key: $X_C':= X_C^{\gamma_C(1)} || X_C^{\gamma_C(2)} || \dots || X_C^{\gamma_C(n)}$.

    \item Bob sends the first $n/2$ blocks of $X_B'$ to Charlie through the authenticated secret channel, together with their positions: $\gamma_B(1), \dots, \gamma_B(n/2)$. The positions can be encoded into $(n/2) \log_2 n$ bits sent over the authenticated secret channel. Likewise, Charlie sends the first $n/2$ blocks of $X_C'$ to Bob, together with the positions $\gamma_C(1), \dots, \gamma_C(n/2)$.
 \end{enumerate}
 
\item \textit{Messaging stage}\quad Alice sends the signed message to Bob, who verifies it and forwards its to Charlie for verification.

\begin{enumerate}[label*=\arabic*.]
    \item The message $Doc$ is signed by $2n$ different hash functions from $\mathcal{F}_{\rm AXU}$, generating $2n$ signatures. Each hash function is obtained from a block of Bob's or Charlie's key. For generating the $j$-th signature  ($1 \leq j \leq n$), Alice generates an LFSR-based Toeplitz matrix $T_{p^j_B,s^j_B}$, i.e. an element of $\mathcal{F}_{\rm AXU}$ (see Appendix~\ref{app:efficient-hashing}), where the initial vector $s^j_B$ is taken from the $j$-th block of Bob's key $X_B$ and the irreducible polynomial $p^j_B$ is obtained by randomly selecting a $b_H$-bit string and by checking that the corresponding polynomial over GF(2), $p^j_B(x)=x^{b_H} + (p^j_B)_{n-1} x^{b_H-1} + \dots + (p^j_B)_1 x + (p^j_B)_0$, is irreducible (see algorithm in Supplementary Material of \cite{QDS-Chen}). Similarly, for signing the $n+j$-th block, Alice generates the Toeplitz matrix $T_{p^j_C,s^j_C}$, where $p^j_C (x)$ is an irreducible polynomial of degree $b_H$ and $s^j_C$ is taken from the $j$-th block of Charlie's key $X_C$.
    
    \item Alice generates the signature $Sig=Sig_1||\dots||Sig_{2n}$ by repeatedly hashing the document and encodes the hashes with the $r^j_B$ and $r^j_C$ keys from Bob's and Charlie's keys. The first $n$ signatures are given by:
    \begin{align}
        Sig_j := (T_{p^j_B,s^j_B}\cdot Doc || p^j_B) \oplus r^j_B \quad j=1,\dots,n,
    \end{align}
    while the following $n$ signatures are given by:
    \begin{align}
        Sig_{j+n} := (T_{p^j_C,s^j_C}\cdot Doc||p^j_C) \oplus r^j_C \quad j=1,\dots,n.
    \end{align}
    Alice sends the pair $\{Doc,Sig\}$ to Bob over a public channel.

    \item Bob discards all the signatures except those that he can legitimately verify, i.e. the $3n/2$ signatures for which he holds the decryption key, namely, $Sig_1, \dots, Sig_n$ and $Sig_{n+\gamma_C(1)}, \dots, Sig_{n+\gamma_C(n/2)}$. For each signature $Sig_j$, with $1 \leq j \leq n$, Bob uses $X^j_B$ to recover the hash of the document and the irreducible polynomial, by computing $Sig_j \oplus r^j_B = (h'_j || {p'}^j_B)$. Then, he verifies the signature by checking whether $h'_j = T_{{p'}^j_B,s^j_B}\cdot Doc$ is satisfied. If not, Bob rejects Alice's signature and sends the symbol $\perp$ to Charlie over an authenticated channel. Analogously, for each signature of the form $Sig_{n + \gamma_C(i)}$, Bob uses the key $X^{\gamma_C(i)}_C$ received from Charlie to recover the hash and the irreducible polynomial: $Sig_{n + \gamma_C(i)} \oplus r^{\gamma_C(i)}_C =(h'_{n+\gamma_C(i)}||{p'}^{\gamma_C(i)}_C)$. Now, Bob verifies that $h'_{n+\gamma_C(i)}=T_{{p'}^{\gamma_C(i)}_C,s^{\gamma_C(i)}_C}\cdot Doc$. If this is not true, Bob aborts and sends $\perp$ to Charlie. If Bob finds no error during verification, he accepts $Doc$ as original and forwards the pair $\{Doc,Sig\}$ to Charlie over an authenticated channel.
    
    \item Charlie verifies the signatures independently of Bob in an analogous way. In particular, he only verifies the signatures $Sig_{n+1}, \dots, Sig_{2n}$ and $Sig_{\gamma_B(1)}, \dots, Sig_{\gamma_B(n/2)}$ for which he holds the decryption keys. If Charlie observes more than $e_{max}$ mismatches out of the $3n/2$ verified signatures, he rejects Alice's signature, otherwise he accepts the document as original. 
\end{enumerate}
\end{enumerate}
\end{protocol}

\subsubsection{Security proof}
We prove the security of the modified protocol that we introduced, Protocol~\ref{QDSprot-INDRA}, against forgery and repudiation. In doing so, we account for the failure probability of the authenticated channel between Bob and Charlie. The proofs are reported in Appendix~\ref{app:security-INDRA}.

\begin{lemma} \label{lm:forgery-INDRA}
    Protocol~\ref{QDSprot-INDRA} is $\varepsilon_{\rm for}$-secure against forgery according to Definition~\ref{def:forgery}, with
    \begin{align}
        \varepsilon_{\rm for} = 
        \Xi(n/2-e_{max},n/2,b_M 2^{1-b_H}) ,
    \end{align}
    where the function $\Xi(k,n,p)$, defined as
    \begin{align}
        \Xi(k,n,p):= \sum_{j=k}^{n} \binom{n}{j} p^j (1-p)^{n-j}, \label{Xi}
    \end{align}
    represents the probability of obtaining at least $k$ successes in $n$ trials with success probability $p$ for each trial.
\end{lemma}

\begin{lemma} \label{lm:repudiation-INDRA}
    Protocol~\ref{QDSprot-INDRA}, with $b_M + 4 n b_H > (n/2)\log_2 n$, is $\varepsilon_{\rm rep}$-secure against repudiation according to Definition~\ref{def:repudiation}, with
    \begin{align}
        \varepsilon_{\rm rep}=\max\left\lbrace \prod_{i=0}^{e_{max}} \frac{n/2 - i}{n-i}, \frac{b_M+4nb_H}{2^{b'_H-1}} \right\rbrace.
    \end{align}
\end{lemma}

\subsection{Amiri \textit{et al.}}  \label{sec:QDS-Andersson}
The authors of Ref.~\cite{USS-andersson} construct an IT-secure QDS scheme with $N+1$ participants, i.e. a sender, $P_0$, and $N$ receivers, $P_1, \dots, P_N$. The protocol is guaranteed to be secure if there is an honest majority of participants and it does not require a fixed party to be honest.

In the $N$-partite scenario, the concept of security against repudiation is more nuanced. Informally, the QDS protocol guarantees that if an honest receiver accepts a signature, then any other other honest receiver accepts it with high probability; this property is called \textit{transferability}. However, there are scenarios where transferability cannot be guaranteed. In these cases, the validity of a signature is collectively established by a majority vote resolution process. A successful repudiation attack would cause the majority vote process to reject a signature that was previously accepted by an honest receiver.

Nevertheless, when analyzing the same tripartite scenario of the other two QDS protocols, i.e. $N=2$, the concepts of transferability and non-repudiation basically reduce to the same requirement, namely, that the sender cannot force the two honest receivers to disagree on the validity of the signature.

\subsubsection{The protocol}

The main idea of the protocol is the following. In the distribution stage, the sender distributes to each receiver a subset of hash functions from the $\mathcal{F}_{\rm ASU}$ family of Sec.~\ref{sec:notation} (see Appendix~\ref{app:efficient-hashing} for more details), through pairwise secret channels. Since this communication is secret, no receiver can play the part of the sender (security against forgery). Afterwards, each receiver partitions the set of received hash functions and sends different partitions to each other receiver through authenticated secret channels. This step shuffles around the set of hash functions held by each receiver, thus preventing the sender from knowing the hash functions held by a given receiver.

In the messaging stage, the sender signs a document with all the hash functions previously distributed to the receivers and appends the corresponding tags. Each receiver counts the number of mismatches between the hash values obtained with their set of hash functions and the corresponding tags, and only accepts the document if the number of mismatches is below threshold. Since the sender does not know the set of hash functions held by a given receiver, they cannot force certain receivers to reject and others to accept the signature (security against repudiation/transferability).

However, in a $N$-user scenario, the sender could collude with a subset of malicious receivers. In that case, the colluding coalition could determine a subset of the hash functions held by a given receiver, thus causing mismatches between the tags and their computed hash values. In order to still guarantee transferability, the authors introduce verification levels, which correspond to different error thresholds: the lower the verification level, the higher the error threshold. The protocol guarantees that if a document is verified by an honest receiver at verification level $l$, any other honest receiver will verify the signature at level $l-1$, with high probability.

We now formally describe the protocol. The protocol fixes a maximum number of dishonest participants it can tolerate, $\omega$, which must be inferior to the majority of participants: $\omega<(N+1)/2$. The maximum fraction of dishonest receivers it can tolerate when colluding with the sender is $d_R$ and is given by: $d_R=(\omega-1)/N$. Finally, the protocol fixes the maximum verification level at $l_{max}$, where $(l_{max}+1)d_R<1/2$.

Importantly, we explicitly specify which protocol steps require authenticated communication, a fact which is not clear from the original formulation of the protocol \cite{USS-andersson}. This helps clarifying the actual amount of resources (e.g. preshared key bits) consumed by the protocol, but also closes potential security loopholes preventing transferability and non-repudiation. \\

\begin{protocol}  \label{QDSprot-Andersson} \caption{QDS protocol \cite{USS-andersson}}
\begin{enumerate}[wide, labelwidth=!, labelindent=0pt]

\item \textit{Distribution stage}\quad The hash functions randomly selected by the sender are distributed to the receivers.

\begin{enumerate}[label*=\arabic*.]
    \item The sender selects uniformly at random (and with replacement) $N^2 k$ hash functions $(f_1,\dots, f_{N^2k})$ from $\mathcal{F}_{\rm ASU}$ (see Appendix~\ref{app:efficient-hashing}), where $k$ is a security parameter.
    \item Each receiver $P_i$ receives via a secret channel the $Nk$ functions: $(f_{(i-1)Nk+1},\dots, f_{iNk})$.
    \item Each receiver $P_i$ randomly partitions the set of indexes corresponding to the received hash functions, $\{(i-1)Nk+1, \dots, iNk\}$, into $N$ sets of size $k$ denoted $R_{i \to 1}, \dots, R_{i \to N}$. The receiver then sends the set $R_{i \to j}$ and the set of corresponding hash functions, $F_{i \to j}:= \{f_r: r \in R_{i \to j}\}$, to $P_j$, using authenticated secret channels [Note that the sets $R_{i \to i}$ and $F_{i \to i}$ are kept by $P_i$]. In this way, each participant $P_i$ holds $Nk$ functions, given by: $F_i:= \cup_{j=1}^N F_{j \to i}$, and their positions, given by: $R_i:= \cup_{j=1}^N R_{j \to i}$.
 \end{enumerate}
 
\item \textit{Messaging stage}\quad The sender sends the signed message and each receiver independently verifies the signature.

\begin{enumerate}[label*=\arabic*.]
    \item The sender sends the pair $\{Doc,Sig\}$ to the desired recipient, $P_i$, over a public channel, where
    \begin{align}
        Sig &:= \left(f_1(Doc), f_2(Doc), \dots, f_{N^2k}(Doc)\right)  \nonumber\\
        &= (t_1, \dots, t_{N^2k}).
    \end{align}
    \item Participant $P_i$ verifies the signature $Sig$ received by the sender for decreasing verification levels $l$, until it either accepts the signature or it runs out of levels ($l=0$ is the last level of verification), in which case the signature is rejected. \\
    Initially, $P_i$ sets the verification level to $l=l_{max}$. For a fixed $j$, participant $P_i$ tests if the $k$ tags generated from the hash functions received by $P_j$ match the tags received by the sender. Thus, receiver $P_i$ defines the test:
    \begin{align}
        T^{Doc}_{i,j,l} = \left\lbrace\begin{array}{ll}
          1   & \mbox{if }\sum_{r \in R_{j \to i}} g(t_r,f_r(Doc)) < s_l k  \\
        0  & \mbox{else}
        \end{array} \right.
    \end{align}
    where $g(x,y)=0$ ($g(x,y)=1$) if $x=y$ ($x\neq y$) and $1/2>s_{-1}>s_0>\dots>s_{l_{max}}>0$ are parameters fixed by the protocol. The test is passed if $T^{Doc}_{i,j,l}=1$.\\
    The participant $P_i$ performs the tests $T^{Doc}_{i,j,l}$ for each $j$ at verification level $l$. The output of the verification at level $l$ is given by:
    \begin{align}
        \mathrm{Ver}_{i,l}(Doc,Sig) = \left\lbrace\begin{array}{ll}
        \mathrm{True}  & \mbox{if } \sum_{j=1}^N T^{Doc}_{i,j,l} >N \delta_l \\
        \mathrm{False} & \mbox{else}
        \end{array} \right.
    \end{align}
    where $\delta_l = 1/2 + (l+1)d_R$. If the verification failed ($\mathrm{Ver}_{i,l}(Doc,Sig)=\mathrm{False}$), the participant $P_i$ decreases by one the verification level and attempts a new verification. The verification process ends when $\mathrm{Ver}_{i,l}(Doc,Sig)=\mathrm{True}$ for some $l \geq 0$, in which case participant $P_i$ accepted the signature at level $l$. Otherwise, $P_i$ rejected the signature. 
    \item If the receiver $P_i$ accepted the signature at some level $l \geq 0$, they forward the pair $\{Doc,Sig\}$ to the next participant via an authenticated channel. Otherwise, the receiver forwards $\perp$ to all remaining participants, flagging that the signature was rejected.
\end{enumerate}
\end{enumerate}
\end{protocol}
\vspace{2ex}

\subsubsection{Security definitions}

Here we report and comment on the security definitions adopted in Ref.~\cite{USS-andersson}. They can be seen as a generalization of Definitions~\ref{def:forgery} and \ref{def:repudiation} presented in Sec.~\ref{sec:notation}.

\begin{definition} \label{def:forgery-Andersson}
    [Forgery]  Let $C \subset \{P_1,\dots,P_N\}$ be a coalition of malicious users, not including the sender, which knows a valid pair $\{Doc,Sig\}$. Then, the QDS protocol is $\varepsilon_{\rm for}$-secure against forgery if, for every forged pair $\{Doc',Sig'\}$ produced by $C$ with $Doc' \neq Doc$, it holds:
    \begin{align}
        \Pr\left[\exists\, P_i \notin C : \,\mathrm{Ver}_{i,0}(Doc',Sig')=\mathrm{True}  \right] \leq \varepsilon_{\rm for}.
    \end{align}
\end{definition}

\begin{definition} \label{def:transferability-Andersson}
    [Non-transferability]  Let $C \subset \{P_0,P_1,\dots,P_N\}$ be a coalition of malicious users including the sender. Then, the QDS protocol is $\varepsilon_{\rm transf}$-secure against non-transferability if, for every pair $\{Doc,Sig\}$ produced by $C$ and every verification level $l \geq 1$, it holds:
    \begin{align}
        \Pr&\left[\exists\, P_i,P_j \notin C : \,\mathrm{Ver}_{i,l}(Doc,Sig)=\mathrm{True} \right.\nonumber\\
        &\left.\quad \wedge   \mathrm{Ver}_{j,l'}(Doc,Sig)=\mathrm{False}\right] \leq \varepsilon_{\rm transf}, \label{nontrasf-condition}
    \end{align}
    where $0 \leq l' < l$.
\end{definition}

We observe that, according to the description of Protocol~\ref{QDSprot-Andersson}, if a signature is verified at level $l$, then it would be verified also at any lower level. That is,
\begin{align}
    &\mathrm{Ver}_{i,l}(Doc,Sig)=\mathrm{True} \implies \nonumber\\
    &\quad\quad\quad \mathrm{Ver}_{i,l'}(Doc,Sig)=\mathrm{True}, \,\, \forall \, l'<l. \label{verfication-cascade}
\end{align}
Thus, it is sufficient to check the condition in \eqref{nontrasf-condition} only for $l'=l-1$.

Moreover, we argue that the transferability definition reported in Definition~\ref{def:transferability-Andersson} does not capture transferability in a meaningful way. Indeed, it only requires that if an honest receiver accepts a message-signature pair $\{Doc,Sig\}$, then another honest receiver accepts the \textit{same} pair $\{Doc,Sig\}$ with high probability. The definition does not account for attacks that may modify the pair $\{Doc,Sig\}$ to be verified by the second receiver. This attack is easily carried out by, e.g., a dishonest receiver who forwards a modified pair to an honest receiver. The protocol has no way to avoid this attack from happening. Alternatively, an attacker could alter the $\{Doc,Sig\}$ pair in the transmission between two honest receivers, by attacking their authenticated channel. We remark that, in the original paper \cite{USS-andersson}, this channel is not explicitly required to be authenticated, hence the attack would always succeed. We remark that the same observation applies to the repudiation definition adopted by Ref.~\cite{USS-andersson} (see below).

When there are disputes caused e.g. by the sender refusing to recognize a signature validated by a legitimate receiver, the participants can invoke a majority vote dispute resolution method, $\mathrm{MV}(Doc,Sig)$, whose outcome is the official accepted outcome. A repudiation attack is successful if the outcome of the majority vote dispute resolution contradicts the outcome of an honest receiver that accepted the signature.

\begin{definition} \label{def:repudiation-Andersson}
    [Repudiation]  Let $C \subset \{P_0,P_1,\dots,P_N\}$ be a coalition of malicious users including the sender. Then, the QDS protocol is $\varepsilon_{\rm rep}$-secure against repudiation if, for every pair $\{Doc,Sig\}$ produced by $C$ and every verification level $l \geq 0$, it holds:
    \begin{align}
        \Pr&\left[\exists\, P_i\notin C : \,\mathrm{Ver}_{i,l}(Doc,Sig)=\mathrm{True} \right.\nonumber\\
        &\left.\quad \wedge \,\mathrm{MV}(Doc,Sig) = \mathrm{Invalid}\right] \leq \varepsilon_{\rm rep}.
    \end{align}
\end{definition}

The $\mathrm{MV}(Doc,Sig)$ method is invoked, for example, when a receiver validates a signature at verification level $l=0$ (this can happen if a malicious coalition forces the failure of all other verification levels, except at $l=0$) and wishes that other receivers would validate the signature as well. The protocol, however, does not guarantee that the next receiver will validate the signature, since transferability is only guaranteed up to level $l=1$ (c.f. Definition~\ref{def:transferability-Andersson}). In this case, the $\mathrm{MV}(Doc,Sig)$ allows the receiver to gain  confirmation from the other users that the signature is authentic. 

The majority vote dispute resolution method is defined following the original protocol \cite{USS-andersson},
\begin{align}
    &\mathrm{MV}(Doc,Sig) = \nonumber\\
    &\left\lbrace \begin{array}{ll}
      \mathrm{Valid}   & \mbox{if } \abs{\{i: \mathrm{Ver}_{i,-1}(Doc,Sig)=\mathrm{True} \}} \geq \floor{N/2} +1 \\
       \mathrm{Invalid}  & \mbox{else,}
    \end{array} \right. \label{MV}
\end{align}
in contrast to the amendment made in Ref.~\cite{QDS-Kiktenko}, where the verification level used in the dispute resolution is $l=0$, that is, $\mathrm{Ver}_{i,-1}$ is replaced by $\mathrm{Ver}_{i,0}$.

In Ref.~\cite{QDS-Kiktenko}, the authors argue that a security loophole arises if the verification is carried out at level $l=-1$, since this case is not guaranteed secure against forgery (c.f. Definition~\ref{def:forgery-Andersson}). A malicious receiver could forge a signature and claim that they verified it at level $l=0$, thus invoking the dispute resolution method, hoping that the signature is accepted by the majority of users if the verification occurs at level $l=-1$. Indeed, Definition~\ref{def:forgery-Andersson} does not guarantee that the protocol can detect forgeries when the forged signature is verified at $l=-1$.

As said, the authors in Ref.~\cite{QDS-Kiktenko} close the loophole by raising the verification level of the dispute resolution method from $l=-1$ to $l=0$. However, in doing so, they lose the original meaning of security against repudiation. Indeed, now it could happen that the first receiver of a pair $\{Doc,Sig\}$ verifies at level $l=0$ and whishes to invoke $\mathrm{MV}(Doc,Sig)$. Since the verification therein occurs at the same level as the first receiver, the protocol cannot guarantee that the dispute resolution will agree with the first user with high probability -- as a matter of fact, transferability is only guaranteed between different verification levels (c.f. Definition~\ref{def:transferability-Andersson}). In other words, security against repudiation, as defined in Definition~\ref{def:repudiation-Andersson}, is lost.

To avoid losing this aspect of the protocol, we solve the loophole issue by computing the probability of success of the above-described forgery attack and by showing that it remains small and of the order of $\varepsilon_{\rm for}$ in the scenarios of interest, thus removing the need for a redefinition of \eqref{MV} (see Lemma~\ref{lm:forgery-Andersson} and the performance analysis in Sec.~\ref{sec:performance}).

\subsubsection{Security proof}

We prove the security of Protocol~\ref{QDSprot-Andersson} with respect to Definitions~\ref{def:forgery-Andersson}, \ref{def:transferability-Andersson}, and \ref{def:repudiation-Andersson}, improving the security parameters where possible while accounting for the failure probability of IT-secure authenticated channels, unlike the original paper \cite{USS-andersson}. In particular, a transferability attack is enabled by altering the hash functions destined to a given participant via successful attacks on the authenticated channels in step 1.3. Note that, in the original paper \cite{USS-andersson}, these channels are not explicitly required to be authenticated.

In the proofs, we assume that the protocol parameters $s_{-1},s_0,\dots,s_{l_{max}}$ are equally spaced in the interval $[0,1/2]$, with spacing $\Delta s=s_{l-1}-s_l$, and we assume that $\Delta s$ is maximal. This choice is optimal to reduce the attack probability on transferability and non-repudiation \cite{USS-andersson}. Moreover, we note that the constraints on the parameters $s_l$ imply the following constraint on the spacing: $(l_{max} + 1)\Delta s <1/2$.

\begin{lemma} \label{lm:forgery-Andersson}
    Protocol~\ref{QDSprot-Andersson} is $\varepsilon_{\rm for}$-secure against forgery according to Definition~\ref{def:forgery-Andersson}, with:
    \begin{align}
        \varepsilon_{\rm for}= (N-\omega)\,\,\Xi(\floor{N/2},N-\omega,p_t),
    \end{align}
    where $p_t$ is defined as:
    \begin{align}
        p_t =\Xi(\floor{k(1-s_0)}+1,k,2^{1-b_H}).
    \end{align}
    Moreover, the forgery attack based on the security loophole discussed in Ref.~\cite{QDS-Kiktenko}, where a dishonest receiver forges the pair $\{Doc',Sig'\}$ and invokes the dispute resolution method for the other parties to accept the pair, succeeds with probability:
    \begin{align}
        p_{\rm attack} = \Xi(\floor{N/2}+1 - \omega ,N-\omega,p_{-1}),   \label{p_attack}
    \end{align}
    with:
    \begin{align}
        p_{-1} = \Xi(\floor{N/2}+1 - \omega ,N-\omega,p_t).
    \end{align}
\end{lemma}

\begin{lemma} \label{lm:transferability&rep-Andersson}
    Protocol~\ref{QDSprot-Andersson} is $\varepsilon_{\rm transf}$-secure against non-transferability according to Definition~\ref{def:transferability-Andersson} and it is $\varepsilon_{\rm rep}$-secure against repudiation according to Definition~\ref{def:repudiation-Andersson}, with:
    \begin{align}
        \varepsilon_{\rm transf}=\varepsilon_{\rm rep}=\binom{N(1-d_R)}{2}  \max\{\varepsilon_{i,j},\varepsilon_{\rm auth},\varepsilon_{\rm hyb}\},
    \end{align}
    where
    \begin{align}
        \varepsilon_{i,j} &= N(1-d_R) \exp\left(- \frac{k}{8(l_{max}+1)^2}\right) \\
        \varepsilon_{\rm auth} &=\left(\frac{ky + k\log_2(Nk)}{2^{b'_H-1}}\right)^{\ceil{N/2 -(l_{max}+1)Nd_R}}
    \end{align}
    \begin{align}
        \varepsilon_{\rm hyb} &=  \frac{ky + k\log_2(Nk)}{2^{b'_H-1}} \nonumber\\
        &\quad\Xi\left(\left\lfloor\frac{N}{2}\right\rfloor +1 , N(1-d_R) ,\exp\left(- \frac{k}{8(l_{max}+1)^2}\right) \right).
    \end{align}
\end{lemma}
The function $\Xi(k,n,p)$, already defined in \eqref{Xi}, represents the probability of at least $k$ successes out of $n$ trials with success probability $p$ for each trial. We present the proofs of Lemmas~\ref{lm:forgery-Andersson} and \ref{lm:transferability&rep-Andersson} in Appendix~\ref{app:security-Andersson}.

\subsection{Comparing approaches to security} \label{sec:security-discussion}

Here we compare the different approaches to security of the three QDS protocols analyzed in this section.

The approach adopted by Protocol~\ref{QDSprot-INDRA} and Protocol~\ref{QDSprot-Andersson} guarantees that every receiver can independently verify the signature in the messaging stage, without the need to communicate with the other receivers. This, however, requires a large communication overhead in the distribution stage.

Indeed, ensuring unforgeability while independently verifying the signature implies that each receiver is verifying a different set of tags, i.e., tags generated from a different set of hash functions applied on the document. At the same time, transferability (non-repudiation) implies that the sender must not know which tags are being verified by each receiver. Therefore, a secure multiparty QDS protocol based on this approach shall employ sufficiently many hash functions while ensuring that the sender cannot associate each receiver to the respective set of hash functions. The result is a resource-consuming distribution stage (each communication must be secret), as described in Protocols~\ref{QDSprot-INDRA} and \ref{QDSprot-Andersson}.

Additionally, as already pointed out in Subsec.~\ref{sec:QDS-Andersson},   the transferability and repudiation definitions adopted by Protocol~\ref{QDSprot-Andersson} could be easily circumvented by equivocation attacks, i.e., by adversaries corrupting the $\{Doc,Sig\}$ pair delivered to a subset of honest receivers. Since the honest receivers have no way to verify that they are testing distinct $\{Doc,Sig\}$ pairs, their disagreement on the signature verification outcome is as severe as a failure of transferability in the standard sense.

Protocol~\ref{QDSprot-Chen} follows instead a different approach, whereby the two receivers are verifying the tag generated by a single hash function, causing a significant reduction in resources consumed by the distribution stage. Moreover, this approach is naturally immune to transferability/repudiation failures caused by equivocation attacks, since all receivers are guaranteed to verify the same tag. Despite its merits, this approach is not immediately generalizable to the multiparty scenario with an arbitrary number of receivers.

\section{Equivocation-resistant signature}  \label{sec:ERS}

In this section, we introduce an efficient and equivocation-resistant signature (ERS) protocol with an arbitrary number of receivers and no trusted authority (patent pending).

The greater efficiency of our ERS protocol stems from generalizing the security approach of the tripartite protocol by Yin \textit{et al.} \cite{QDS-Chen} (analyzed as Protocol~\ref{QDSprot-Chen} in this manuscript) to the multiparty scenario. In particular, all receivers end up verifying a single tag generated by one hash function.

To achieve unforgeability in this case, the information on the employed hash function is split among the receivers in the distribution stage. Then, in the messaging stage, additional public communication between receivers ensures that all receivers commit to the same document-signature pair before collectively reconstructing the employed hash function, allowing verification to begin. The receiver consensus on the document-signature pair to be verified is achieved via Bracha's reliable broadcast (BRB) protocol \cite{ReliableBroadcast-Bracha} (see Appendix~\ref{app:BRB}), thus resulting inherently immune to equivocation-based attacks. As a consequence, our ERS protocol satisfies a stronger notion of transferability, which we report in Definition~\ref{def:transferability-ERS}, compared to the standard one adopted by multiparty QDS protocols (i.e.,  Definition~\ref{def:transferability-Andersson}).

Before presenting the protocol, we set the notation. The set of all receivers is denoted $\mathcal{P}=\{P_1,\dots,P_N\}$. For signing documents in our ERS protocol, we employ universal hashing functions from the $\mathcal{F}_{\rm ASU}$ family (c.f. Table~\ref{tab:efficient-hashing}), where we indicate with $h$ the string of $y$ bits uniquely identifying the function $f_h \in \mathcal{F}_{\rm ASU}$ (see Appendix~\ref{app:efficient-hashing} for details). Meanwhile, in agreement with the rest of the paper, all authenticated channels (including those used in the BRB protocol) are implemented via WC authentication with key recycling, where the hash family is $\mathcal{F}_{\rm AXU}$ from Table~\ref{tab:efficient-hashing} and the hashes are $b'_H$ bits long.\\

\begin{protocol} \label{ERSprot} \caption{Equivocation-resistant signature (ERS)}
\begin{enumerate}[wide, labelwidth=!, labelindent=0pt]

\item \textit{Distribution stage}\quad The sender establishes a pairwise secret key with each receiver, labeled $X_i$ for receiver $P_i$, with $i=1,\dots, N$. Each key is $y$ bits long. The parameter $y$ represents the number of bits required to specify an element of $\mathcal{F}_{\rm ASU}$ and is given in Appendix~\ref{app:efficient-hashing}.
 
\item \textit{Messaging stage}

\begin{enumerate}[label*=\arabic*.]
    \item The sender selects uniformly at random one hash function, $f_h$, from the set $\mathcal{F}_{\rm ASU}$. Let $f_h$ be uniquely identified by the bitstring $h$.
    \item The sender sends the pair $\{Doc,Sig\}$ to the desired receiver over a public channel, where
    \begin{align}
        Sig &:= \left(f_h(Doc)||  h \oplus X_s \right) \nonumber\\
        & \equiv \left(t || \tilde{h}_s \right),
    \end{align}
    where $X_s =\oplus_{i=1}^N X_i$. Let $P_1$ be the designated receiver, without loss of generality.
    \item Receiver $P_1$ broadcasts the pair $\{Doc,Sig\}$ to all parties in $\mathcal{P}$ via the BRB protocol. If the BRB protocol does not abort, let $\{Doc_i,Sig_i\}$ be the pair held by receiver $P_i$ at the end of the BRB protocol, with $Sig_i = (t_i || \tilde{h}_i)$. Otherwise, the ERS protocol aborts.
    \item Each receiver $P_i$ ($i \neq 1$) sends their secret key $X_i$ to $P_1$ over an authenticated channel.
    \item After computing the XOR of the received keys and of the key $X_1$, thus obtaining $X_r$, receiver $P_1$ broadcasts the string $X_r$ to the parties in $\mathcal{P}$ via the BRB protocol. If the BRB protocol does not abort, let $X_{r,i}$ be the string obtained by $P_i$ at the end of the BRB protocol. Otherwise, the ERS protocol aborts.
    \item Each receiver $P_i$ independently retrieves the hash function used by the sender by computing: $h_i := \tilde{h}_i \oplus X_{r,i}$. Let $f_{h_i}\in \mathcal{F}_{\rm ASU}$ be the hash function uniquely associated to the string $h_i$. The receiver accepts the signature if $f_{h_i}(Doc_i) = t_i$ and rejects it otherwise. The outcome of the ERS protocol for receiver $P_i$ reads:
    \begin{align}
        \mathrm{Ver}_{i}(Doc_i,Sig_i) = \left\lbrace\begin{array}{ll}
        \mathrm{True}  & \mbox{if } f_{h_i}(Doc_i) = t_i \\
        \mathrm{False} & \mbox{if } f_{h_i}(Doc_i) \neq t_i \\
        \mathrm{Abort} & \mbox{if abort in step 2.3 or 2.5}
        \end{array} \right.
    \end{align}
\end{enumerate}
\end{enumerate}
\end{protocol}
\vspace{2ex}

\subsection{Security proof}

Our ERS protocol is proved secure according the following definitions of unforgeability and transferability (or non-repudiation). These are to be understood as an extension of the definitions presented in the tripartite scenario (Definitions~\ref{def:forgery} and \ref{def:repudiation}) and, at the same time, a replacement of the weaker definitions adopted by Amiri \textit{et al.} in Ref.~\cite{USS-andersson} (c.f. Definitions~\ref{def:forgery-Andersson}, \ref{def:transferability-Andersson}, and \ref{def:repudiation-Andersson}).

\begin{definition} \label{def:forgery-ERS}
    [Unforgeability]  Let $C \subset \{P_1,\dots,P_N\}$ be a coalition of malicious receivers which knows a valid $\{Doc,Sig\}$ pair. Then, the ERS protocol is $\varepsilon_{\rm for}$-secure against forgery if, for every forged pair $\{Doc',Sig'\}$ produced by $C$ with $Doc' \neq Doc$, it holds:
    \begin{align}
        \Pr\left[\exists\, P_i \notin C : \,\mathrm{Ver}_{i}(Doc',Sig') =\mathrm{True}  \right] \leq \varepsilon_{\rm for}.
    \end{align}
\end{definition}

\begin{definition} \label{def:transferability-ERS}
    [Transferability/non-repudiation]  Let $C \subset \{P_1,\dots,P_N\}$ be a coalition of malicious receivers that can collude with the sender. Then, the ERS protocol is $\varepsilon_{\rm rep}$-secure against repudiation (non-transferability) if it holds:
    \begin{align}
        \Pr&\left[\exists\, P_i,P_j \notin C : \,\mathrm{Ver}_{i}(Doc_i,Sig_i) \neq \mathrm{Ver}_{j}(Doc_j,Sig_j)\right] \nonumber\\
        &\leq \varepsilon_{\rm rep}, 
    \end{align}
    for any deviation of the actions of $C$ from the honest behaviour.
\end{definition}

We observe that the transferability notion in Definition~\ref{def:transferability-ERS} is stronger than that in Definition~\ref{def:transferability-Andersson} from Ref.~\cite{USS-andersson}. Indeed, it requires that either all honest receivers agree on the outcome of signature verification, or all honest receivers abort the protocol. We emphasize that this natural feature is not guaranteed by Protocol~\ref{QDSprot-Andersson}, since dishonest parties can always modify the $\{Doc,Sig\}$ pair received from an honest receiver before forwarding it to another honest receiver, forcing the two honest receivers to disagree. Moreover, the transferability of a signature in Protocol~\ref{QDSprot-Andersson} is also limited by the system of verification levels. If an honest receiver verifies the signature at level $l=0$ in Protocol~\ref{QDSprot-Andersson}, they are not guaranteed that other receivers will agree with their outcome. This is mitigated by invoking a majority vote dispute resolution, which however requires additional communication rounds among all receivers and constrains the number of dishonest receivers colluding with the sender to be strictly smaller than $(N-1)/2$.

Conversely, our ERS protocol presents no limits on signature transferability and is proved secure according to the stronger notion of transferability, Definition~\ref{def:transferability-ERS}. As such, our protocol intrinsically guarantees the global agreement of honest receivers in every scenario, thus not requiring a dispute resolution method.

\begin{lemma} \label{lm:trasferability-ERS}
    Protocol~\ref{ERSprot} is $\varepsilon_{\rm rep}$-secure against repudiation (or non-transferability) according to Definition~\ref{def:transferability-ERS}, with $\abs{C}<N/3$ and
    \begin{align}
        \varepsilon_{\rm rep}= (b_M+b_H + y)/2^{b'_H-1}.
    \end{align} 
\end{lemma}

The fundamental restriction on the size of the colluding set, namely $|C|<N/3$, is imposed by the security of the BRB protocol and cannot be improved in the context of our security model \cite{consensus_bound}.

Moreover, our ERS protocol is secure against forgery according to Definition~\ref{def:forgery-ERS}, with an arbitrary number of malicious receivers. Vice-versa, Protocol~\ref{QDSprot-Andersson} requires an honest majority of receivers.

\begin{lemma} \label{lm:forgery-ERS}
    Protocol~\ref{ERSprot} is $\varepsilon_{\rm for}$-secure against forgery according to Definition~\ref{def:forgery-ERS}, with $C$ being an arbitrary set of colluding receivers and
    \begin{align}
        \varepsilon_{\rm for}= 2^{1-b_H}.
    \end{align}
\end{lemma}

We present the proofs of Lemma~\ref{lm:trasferability-ERS} and Lemma~\ref{lm:forgery-ERS} in Appendix~\ref{app:security-ERS}.

\section{Performance benchmarking} \label{sec:performance}

\begin{table*}[t]
\renewcommand\arraystretch{2.5}
\setlength{\tabcolsep}{3pt}
\centering
\begin{tabular}[t]{|p{0.04\linewidth}|p{0.10\linewidth}|p{0.25\linewidth}|p{0.06\linewidth}|p{0.23\linewidth}|p{0.24\linewidth}|}
\toprule
$\bm{N}$ & \textbf{Protocol} & $\bm{\ell}_{\bm{P}}$ & $\bm{\ell}_{\bm{S}}$ & $\bm{\varepsilon}_{\rm\bf for}$ & $\bm{\varepsilon}_{\rm\bf rep}$ \\
\midrule
$>2$ & Protocol~\ref{ERSprot} (ERS prot.) & $y + 13(N-1)b'_H$ & $b_H+y$ & $2^{1-b_H}$ & $(b_M+b_H + y)/2^{b'_H-1}$ \\
$>2$ & Protocol~\ref{QDSprot-Andersson} of Ref.~\cite{USS-andersson} & $(2N-1) k y + k (N-1)\log_2(Nk) + (N-1)6 b'_H + b'_H$ & $N^2 k b_H$ & $(N-\omega_R -1)\,\,\Xi(\floor{N/2},N-\omega_R -1,p_t)$ & $\binom{N -\omega_R}{2}  \max\{\varepsilon_{i,j},\varepsilon_{\rm auth},\varepsilon_{\rm hyb}\}$ \\
\midrule
$2$ & Protocol~\ref{ERSprot} (ERS prot.) & $y + 5 b'_H$ & $b_H+y$ & $2^{1-b_H}$ & $(b_M+b_H + y)/2^{b'_H-1}$ \\
$2$ & Protocol~\ref{QDSprot-Chen} of Ref.~\cite{QDS-Chen}  & $3 b_H + 5b'_H$ & $2b_H$ & $b_M/2^{b_H-1}$ & $(2b_H+b_M)/2^{b'_H-1}$ \\
$2$ & Protocol~\ref{QDSprot-INDRA} mod. from Ref.~\cite{QDS-INDRA}  & $6 n b_H + n \log_2 n + 7b'_H$ & $4nb_H$ & $\Xi(n/2-e_{max},n/2,b_M 2^{1-b_H})$ & $\max\left\lbrace \prod_{i=0}^{e_{max}} \frac{n/2 - i}{n-i}, \frac{b_M+4nb_H}{2^{b'_H-1}} \right\rbrace$ \\
$2$ & Protocol~\ref{QDSprot-Andersson} of Ref.~\cite{USS-andersson}  & $3ky + k\log_2(2k) +7 b'_H$ & $4k b_H$ & $\Xi\left(\left\lfloor\frac{3}{4}k \right\rfloor +1,k,2^{1-b_H}\right)$ & $\max \left\lbrace 2 e^{-k/32}, \frac{ky + k\log_2(2k)}{2^{b'_H-1}} \right\rbrace$ \\
\bottomrule
\end{tabular}
\caption{Main figures of merit for the analyzed QDS protocols in the scenarios of one sender and $N$ receivers, ($\omega_R$  malicious receivers when $N>2$). We report the security parameter against forgery ($\varepsilon_{\rm for}$) and repudiation  ($\varepsilon_{\rm rep}$), as well as the consumed preshared secret key bits per receiver ($\ell_P$) and the length of the signature ($\ell_S$). Each protocol is used to sign a $b_M$-bit document with hash functions drawn from either the $\mathcal{F}_{\rm AXU}$ family (Protocols~\ref{QDSprot-Chen} and \ref{QDSprot-INDRA}) or the $\mathcal{F}_{\rm ASU}$ family (Protocol~\ref{QDSprot-Andersson} and Protocol~\ref{ERSprot}), producing hashes of $b_H$ bits. The parameter $y$ represents the number of bits required to specify an element of $\mathcal{F}_{\rm ASU}$ and is given in Appendix~\ref{app:efficient-hashing}. The function $\Xi$ is given in Eq.~\eqref{Xi} while $p_t$, $\varepsilon_{i,j}$, $\varepsilon_{\rm auth}$, and $\varepsilon_{\rm hyb}$ are reported in Lemmas~\ref{lm:forgery-Andersson} and \ref{lm:transferability&rep-Andersson}. Moreover, $b'_H$ is the bit-length of the tags used for WC authentication and $k,n,e_{max}$ are tunable security parameters.}\label{tab:protocols_params}
\end{table*}

In this section we benchmark the performance of the ERS protocol (Protocol~\ref{ERSprot}) against the three promising QDS protocols analyzed in Sec.~\ref{sec:all-QDS} (Protocols~\ref{QDSprot-Chen}, \ref{QDSprot-INDRA}, and \ref{QDSprot-Andersson}) and show its superior efficiency in terms of consumed resources.

Since Protocol~\ref{QDSprot-Andersson} is the only other protocol analyzed in this paper that is designed for an arbitrary number of receivers, we separately compare Protocol~\ref{ERSprot} with Protocol~\ref{QDSprot-Andersson} in the multiparty scenario with $N>2$ receivers. Afterwards, we benchmark the ERS protocol against all three protocols in the tripartite scenario with $N=2$ receivers.

To achieve a fair comparison, we independently optimize each protocol to maximize its efficiency, i.e., reduce the consumption of preshared secret bits. Specifically, given a document of fixed length ($b_M$), we numerically optimize the protocol's free parameters, such that the number of preshared secret key bits per receiver ($\ell_P$) is minimized, under the constraint that $\varepsilon_{\rm rep} + \varepsilon_{\rm for} \leq 10^{-10}$. We observe that optimizing the value of $\ell_P$ also implicitly reduces the length of the signature used by the protocol ($\ell_S$).

In particular, Protocol~\ref{QDSprot-Chen} (based on the protocol by Yin \textit{et al.} \cite{QDS-Chen}) is optimized over the parameters $b_H,b'_H$. Protocol~\ref{QDSprot-INDRA} (a modified version of the protocol by García Cid \textit{et al.} \cite{QDS-INDRA}) is optimized over $n,b_H,b'_H, e_{max}$. Protocol~\ref{QDSprot-Andersson} (based on the protocol by Amiri \textit{et al.} \cite{USS-andersson}) is optimized over $k,b_H,b'_H$. And Protocol~\ref{ERSprot} is optimized over $b_H$ and $b'_H$. We recall that  every authenticated channel is implemented via WC authentication with key recycling and tags of $b'_H$ bits, while $b_H$ is the length of the document hashes.

We observe that Protocol~\ref{QDSprot-Andersson} can be additionally optimized over the parameters $s_{-1},s_0, s_1,$ etc. However, as argued in Subsec.~\ref{sec:QDS-Andersson}, the spacing ($\Delta s$) between the parameters is optimal when constant and maximal. Thus, in our optimizations we introduce a vanishing parameter $\delta_s>0$ such that the constraints $1/2 > s_{-1}$ and $s_{l_{max}}>0$ are respected, while maximizing the spacing $\Delta s$. We have:
\begin{align}
    s_{-1} &= \frac{1}{2} -\delta_s \\
    s_{l_{max}} &= \delta_s \\
    \Delta s &= \frac{1/2 - 2 \delta_s}{l_{max} + 1}. \label{Delta_s}
\end{align}
Thus, the value of $s_0$ is fixed to:
\begin{align}
    s_0 &= s_{-1} - \Delta s \nonumber \\
    &=  \delta_s + \left(\frac{1}{2} - 2\delta_s \right)\frac{l_{max}}{l_{max} + 1}. \label{s0}
\end{align}

In Table~\ref{tab:protocols_params} we report the security parameters, the number of preshared key bits per receiver and the signature length of each protocol that are used in our optimizations. In the following, we illustrate how to obtain the number of preshared key bits reported for each protocol in Table~\ref{tab:protocols_params}. Note that we assume that every secret channel is implemented via OTP, thus consuming as many preshared key bits as the length of the exchanged messages.

\textbf{Protocol~\ref{QDSprot-Chen}}\quad The protocol requires each receiver to hold a preshared secret key of length $3 b_H$. Moreover, to establish the authenticated channel between Bob and Charlie (i.e., to agree on a function from $\mathcal{F}_{\rm AXU}$), each receiver consumes $2b'_H$ preshared bits. The three messages that are sent over the authenticated channel amount to additional $3b'_H$ consumed preshared bits (needed to encrypt the tags). These contributions sum to: $\ell_P = 3 b_H + 5 b'_H$.

\textbf{Protocol~\ref{QDSprot-INDRA}}\quad The protocol requires each receiver to share a $3nb_H$-bit secret key with Alice. Then, it consumes $3nb_H + n \log_2 n$ preshared bits to secretly exchange half of the key blocks between Bob and Charlie via OTP, together with their positions. To establish the authenticated channel between Bob and Charlie, each receiver consumes $2b'_H$ preshared bits. Bob and Charlie exchange $4$ authenticated messages in the distribution stage; moreover, Bob forwards the $\{Doc,Sig\}$ pair on the authenticated channel in the messaging stage, for a total of $5$ authenticated messages. This gives us: $\ell_P = 6 n b_H + n \log_2 n + 7b'_H$.

\textbf{Protocol~\ref{QDSprot-Andersson}}\quad To carry out step 1.2 of Protocol~\ref{QDSprot-Andersson}, each receiver must hold $Nky$ preshared key bits with the sender, where $y$ is the number of bits required to specify one element of the $\mathcal{F}_{\rm ASU}$ family and is provided in Appendix~\ref{app:efficient-hashing}. For step 1.3 of the protocol, each receiver $P_i$ sends to another receiver, $P_j$, $k$ hash functions from $\mathcal{F}_{\rm ASU}$ and their positions, thus consuming a total of $(N-1)(ky + k\log_2 (Nk))$ preshared secret bits. Moreover, this communication is also authenticated. Hence, it requires each receiver to share, with any other receiver, $2b'_H$ bits to agree on the hash function and additional $4b'_H$ bits to exchange the hash functions and their positions. Finally, $b'_H$ preshared bits are consumed to forward the $\{Doc,Sig\}$ pair to the next recipient in an authenticated manner. The resulting number of preshared key bits reads: $\ell_P=N k y + (N-1)(ky+k\log_2(Nk)) + (N-1)6 b'_H + b'_H$.

\textbf{Protocol~\ref{ERSprot}}\quad Each receiver holds a preshared secret key of $y$ bits that is consumed by the protocol, which thus contributes to the total of $\ell_P$ preshared bits consumed per receiver. Moreover, we observe that receiver $P_1$ is endowed with a special role as the broadcaster of the BRB protocols. This implies that $P_1$ consumes more preshared bits than the other receivers. To avoid privileging our ERS protocol in the performance comparison, we compute $\ell_P$ for the worst-case receiver, i.e. for $P_1$. Since $P_1$ needs to communicate via authenticated channels with each other receiver, it consumes $2b'_H (N-1)$ bits to set up the channels. In each BRB protocol, $P_1$ exchanges a total of $5(N-1)$ messages (sent plus received). In addition, $P_1$ receives $X_i$ over an authenticated channel from each other receiver, for a total of $N-1$ messages. Summing everything up, we obtain: $\ell_P=y+13(N-1)b'_H$.

\textbf{Protocol~\ref{ERSprot} ($N=2$)}\quad In the tripartite scenario we can simplify the ERS protocol by observing that the BRB protocol becomes superfluous. Indeed, the purpose of the BRB protocol in the $N>2$ case is to ensure that all receivers obtain the same $\{Doc,Sig\}$ pair, which is necessary to satisfy our transferability definition (Definition~\ref{def:transferability-ERS}). However, for $N=2$, there is only one other receiver other than $P_1$, hence no equivocation is possible and the BRB protocol in steps 2.3 and 2.5 can be replaced by simple authenticated messages. By doing so, the computation of $\ell_P$ reduces to: $\ell_P = y + 5 b'_H$, where we replaced the contribution from the two BRB protocols with the contribution given by two authenticated messages, i.e., $2b'_H$.

\subsection{Multiparty scenario}

Here we benchmark our ERS protocol, Protocol~\ref{ERSprot}, against Protocol~\ref{QDSprot-Andersson} of Ref.~\cite{USS-andersson} for a wide range of document sizes, $b_M \in [10^2,10^{10}]$, and for networks with up to $N=102$ participants other than the sender.

In our optimizations we define $\omega_R$ to be the number of dishonest receivers. Recall that, for both protocols, transferability/repudiation attacks allow for collusion between dishonest receivers and a dishonest sender, meaning that the total number of dishonest participants is $\omega=\omega_R+1$. In our simulations, we choose $\omega_R = \floor{N/4}$, which satisfies the constraint on the number of dishonest receivers of Protocol~\ref{QDSprot-Andersson} ($\omega_R< (N-1)/2$) and of Protocol~\ref{ERSprot} ($\omega_R< N/3$).

Moreover, with our choice for $\omega_R$, the maximum verification level of Protocol~\ref{QDSprot-Andersson} is $l_{max}=1$ for most values of $N$, meaning that the signature can be transferred between receivers at least one time. Indeed, in Protocol~\ref{QDSprot-Andersson} $l_{max}$ is constrained to satisfy $(l_{max}+1) \omega_R < N/2$. Since $l_{max}$ must be a positive integer and since we want to grant the largest possible number of transfers between receivers, we choose $l_{max}$ to be the largest integer that satisfies the constraint. Thus, we have:
\begin{align}
     l_{max} (\omega_R) = \left\lceil\frac{N}{2 \omega_R} -1 \right\rceil -1. \label{lmax}
\end{align}
Then, for our choice of $\omega_R$, we have:
\begin{align}
     l_{max} ( \floor{N/4}) = \left\lbrace \begin{array}{ll}
      2   & \mbox{for } N=7 \\
       0  & \mbox{for $N$ multiple of 4} \\
       1  & \mbox{else.}
    \end{array} \right.
\end{align}

\begin{figure}[h!]
    \centering
    \includegraphics[width=1\linewidth,keepaspectratio]{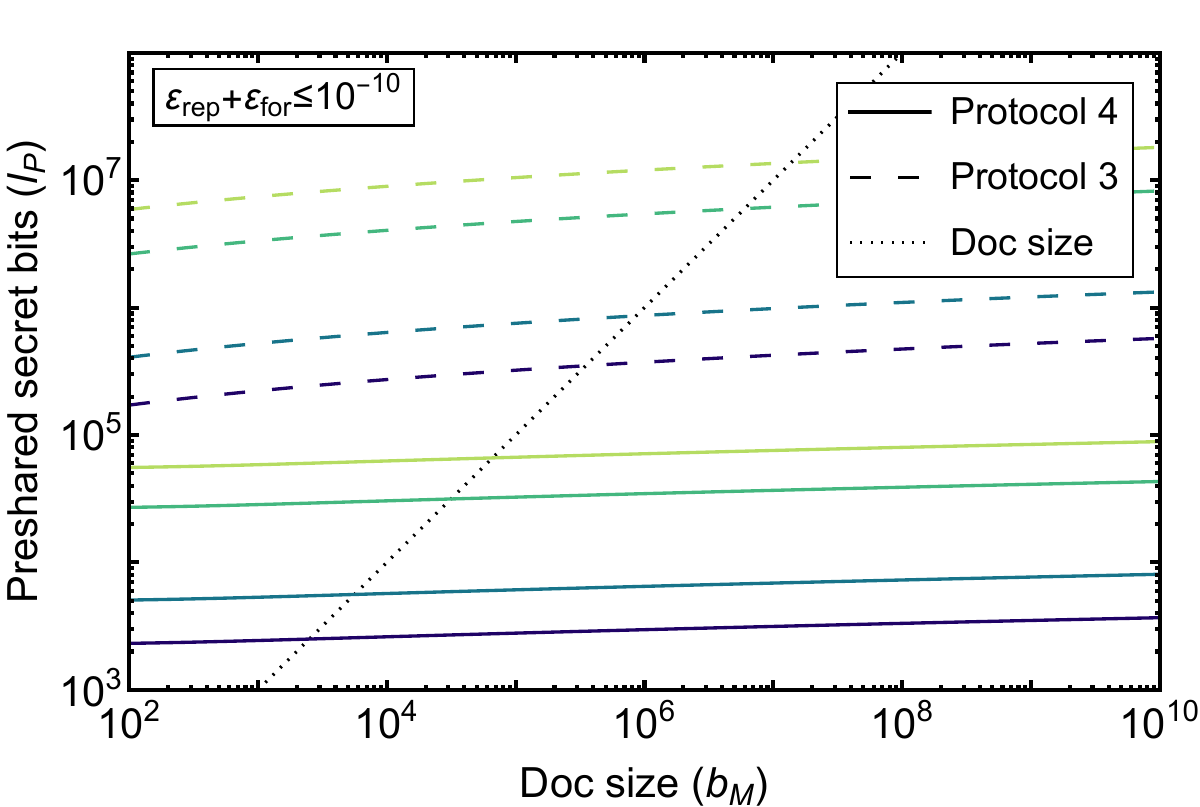}
    \includegraphics[width=1\linewidth,keepaspectratio]{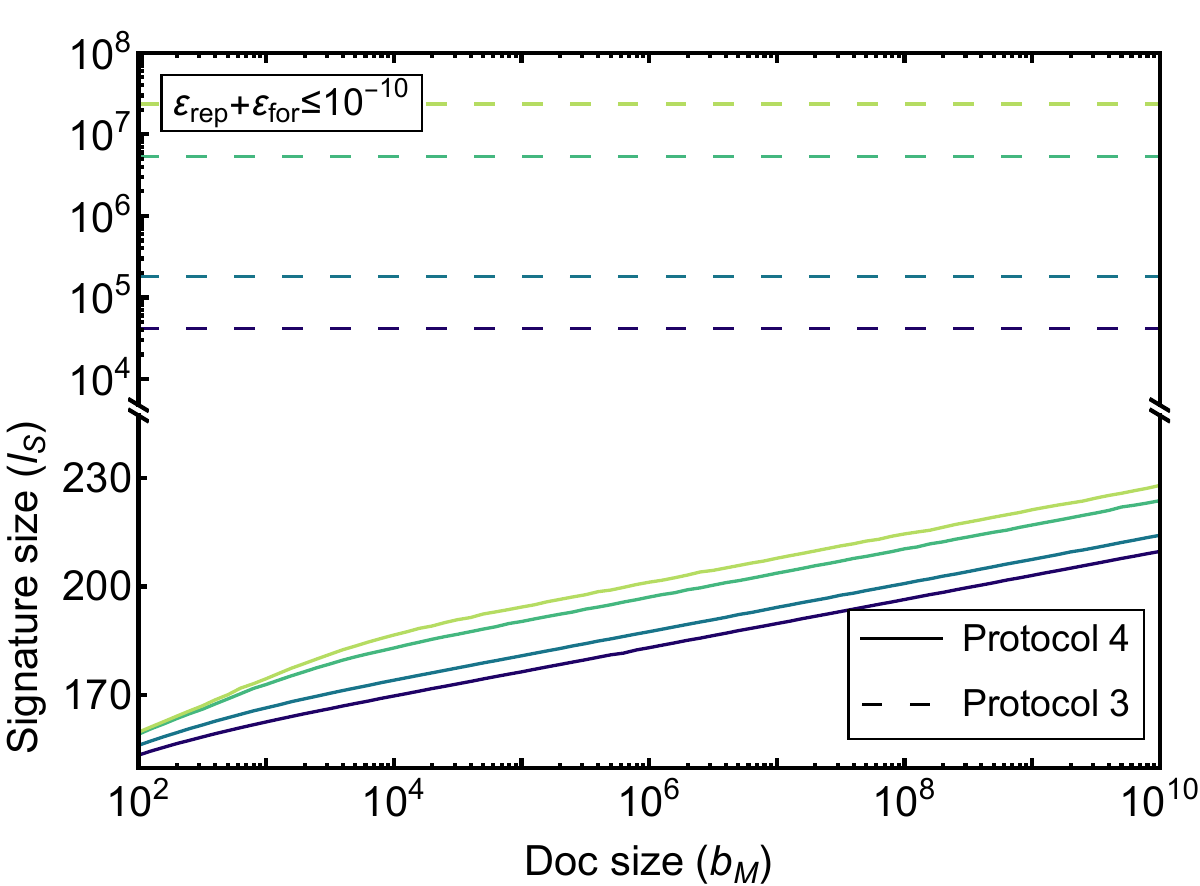}
    \caption{The optimal number of preshared secret bits per receiver ($\ell_P$) as a function of the document size ($b_M$) for our ERS protocol (Protocol~\ref{ERSprot}) and for the protocol by Amiri \textit{et al.} \cite{USS-andersson} (Protocol~\ref{QDSprot-Andersson}), when minimized under the constraint $\varepsilon_{\rm rep} + \varepsilon_{\rm for} \leq 10^{-10} $. We also report the corresponding signature lengths ($\ell_S$). Progressively lighter shades denote scenarios with increasing numbers of receivers: $N=5$, $N=10$, $N=50$, and $N=102$.} \label{fig:preshared_signature-ERS}
\end{figure}

\subsubsection{Results}
The results of the numerical optimizations of Protocol~\ref{QDSprot-Andersson} and Protocol~\ref{ERSprot} are presented in Fig.~\ref{fig:preshared_signature-ERS} for a varying number of receivers, namely $N=5,10,50$ and $N=102$, while the number of dishonest receivers is fixed at $\omega_R = \floor{N/4}$ and it is such that both protocols can guarantee security.

From the top plot in Fig.~\ref{fig:preshared_signature-ERS}, we observe that our ERS protocol consumes less preshared secret bits than Protocol~\ref{QDSprot-Andersson} for all document sizes and all values of $N$, consistently outperforming Protocol~\ref{QDSprot-Andersson} by more than two orders of magnitude. For example, to sign a $10^6$-bit document and share it with $N=102$ receivers, our protocol consumes $7 \cdot 10^4$ preshared bits while Protocol~\ref{QDSprot-Andersson} requires $1.2 \cdot 10^7$ preshared bits. In terms of scaling, for both protocols the number of preshared secret bits scales linearly with the number of receivers: $\ell_P \sim N$.

The performance improvement of our ERS protocol is even more striking when looking at the signature lengths in the bottom plot of Fig.~\ref{fig:preshared_signature-ERS}. Indeed, our protocol requires signatures that scale with the logarithm of the document size $b_M$ and that never exceed $230$ bits. Moreover, the signature length is almost unaffected by the change in the number of receivers, with variations of at most $18$ bits between cases $N=5$ and $N=102$. Indeed, the optimal parameters $b_H$ and $b'_H$ are basically constant in $N$, with small variations due to the increased weight of $b'_H$ in $\ell_P$ as $N$ increases. Conversely, Protocol~\ref{QDSprot-Andersson}'s signatures are constant with respect to $b_M$ but scale quadratically in $N$, thus requiring more than $10^7$ bits for $N=102$, which is five orders of magnitude larger than the size of the corresponding signature of the ERS protocol. 

We recall that Protocol~\ref{QDSprot-Andersson} is vulnerable to a forgery attack based on the loophole noted in Ref.~\cite{QDS-Kiktenko} and discussed in Subsec.~\ref{sec:QDS-Andersson}. We provide the success probability of said forgery attack in Eq.~\eqref{p_attack}. By computing the success probability for $N=5,10,50$ and $N=102$ and $\omega_R = \ceil{N/4}$, we observe that $p_{\rm attack} \ll \varepsilon_{\rm for}$, implying that the attack is actually much less likely to succeed than regular forgery attacks.

We conclude that our ERS protocol significantly improves on the performance metrics of Protocol~\ref{QDSprot-Andersson} in $N$-receiver scenarios, both in terms of preshared bits consumption and signature sizes.\\

\subsection{Tripartite scenario}

\begin{figure}[htb]
    \centering
    \includegraphics[width=1\linewidth,keepaspectratio]{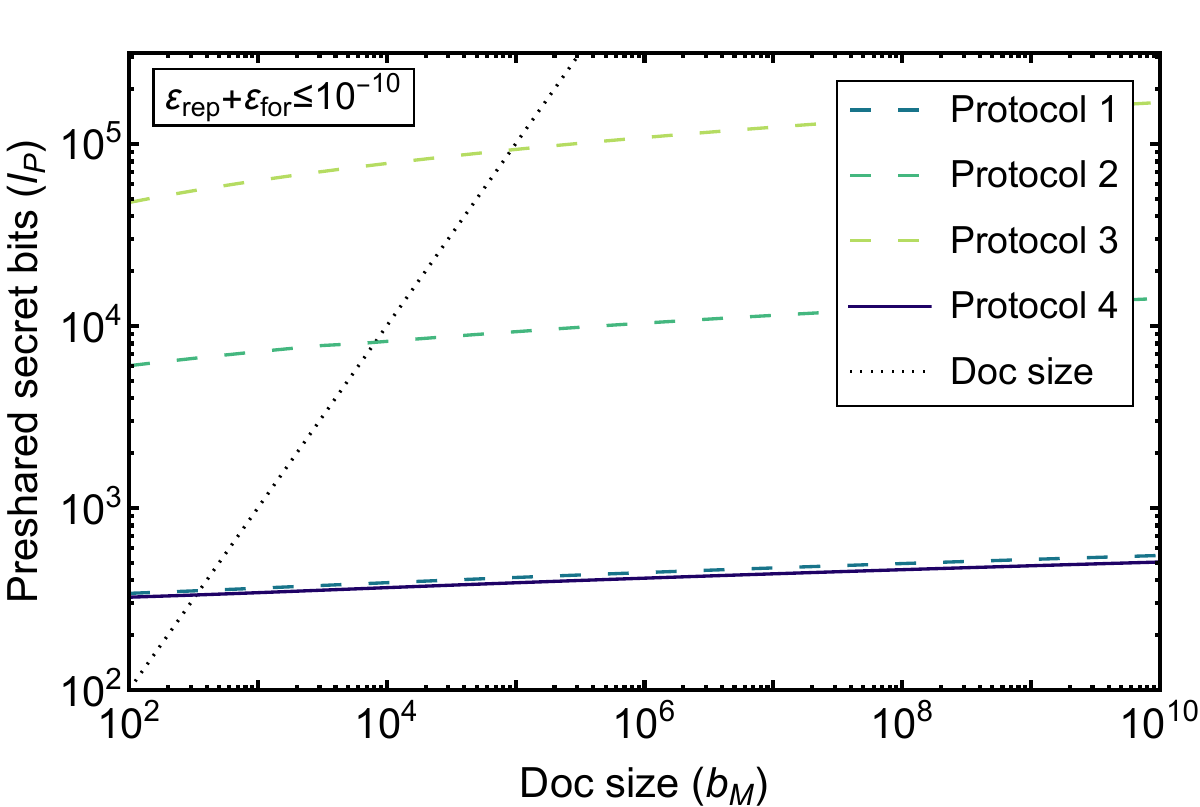}
    \includegraphics[width=1\linewidth,keepaspectratio]{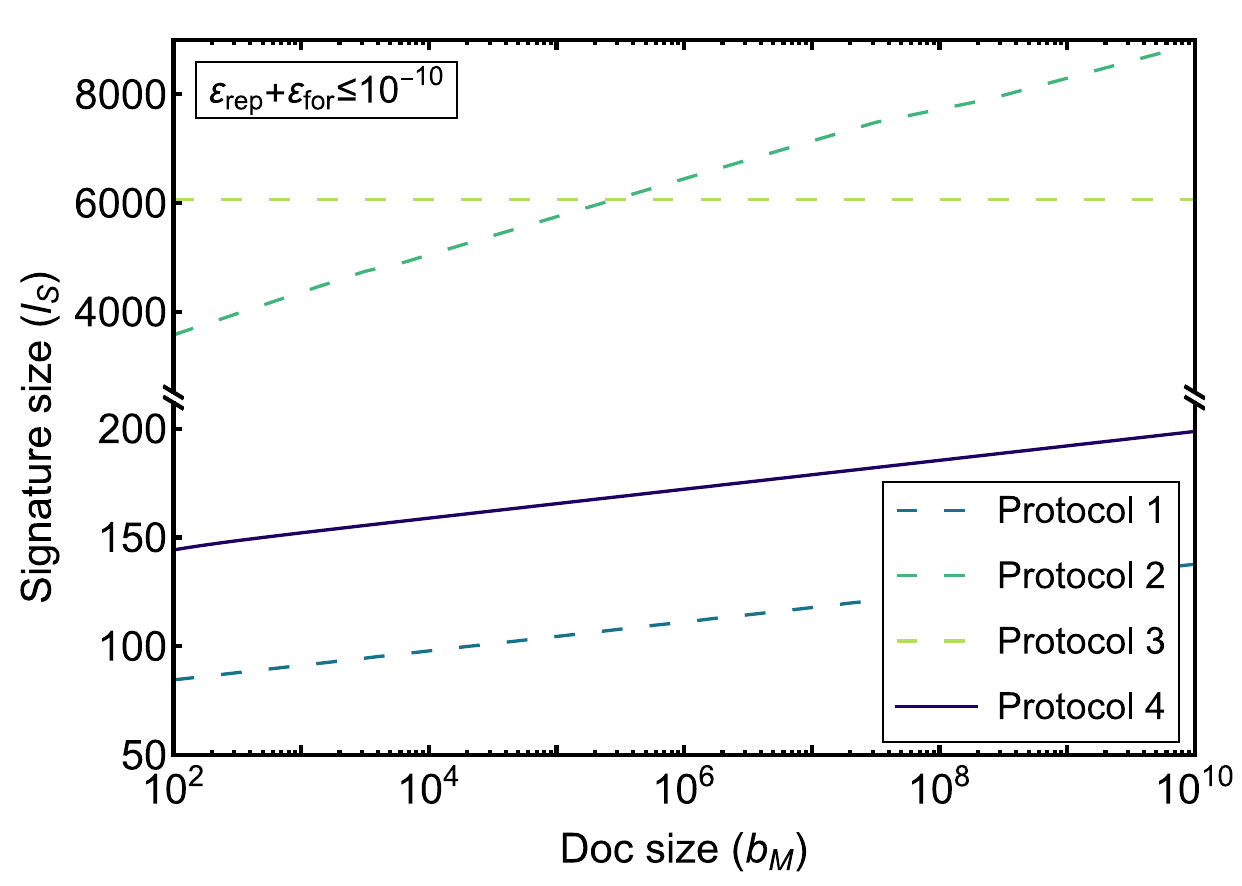}
    \caption{The optimal number of preshared secret bits per receiver ($\ell_P$) as a function of the document size ($b_M$) when minimized under the constraint $\varepsilon_{\rm rep} + \varepsilon_{\rm for} \leq 10^{-10}$, in the tripartite scenario of one sender and two receivers. We report the plot lines for our ERS protocol (Protocol~\ref{ERSprot}) and for the three QDS protocols by Yin \textit{et al.} \cite{QDS-Chen} (Protocol~\ref{QDSprot-Chen}), García Cid \textit{et al.} \cite{QDS-INDRA} (Protocol~\ref{QDSprot-INDRA}), and Amiri \textit{et al.} \cite{USS-andersson} (Protocol~\ref{QDSprot-Andersson}) that are analyzed in Sec.~\ref{sec:all-QDS}. We also report the corresponding signature lengths ($\ell_S$).}
    \label{fig:allprotocols-tripartite}
\end{figure}

Here we focus on the tripartite scenario of one sender and two receivers ($N=2$) and we compare the ERS protocol with the three QDS protocols analyzed in Sec.~\ref{sec:all-QDS}.

In this case, none of the four protocols allow a receiver to collude with the sender in a repudiation attack ($\omega_R =0$) and the maximum number of signature transfers in Protocol~\ref{QDSprot-Andersson} is $l_{max}=1$. Thus, $s_0 = 1/4$ according to Eq.~\eqref{s0}. With these parameters, the success probability ($p_{\rm attack}$) of the forgery attack on Protocol~\ref{QDSprot-Andersson} noted in Ref.~\cite{QDS-Kiktenko} reduces to the regular forgery probability: $p_{\rm attack} = \varepsilon_{\rm for}$. Thus, the forgery attack described in Ref.~\cite{QDS-Kiktenko} is not more impactful than any other forgery attack in the tripartite case.

The results of our numerical optimizations for the four protocols are presented in Fig.~\ref{fig:allprotocols-tripartite} for documents of $b_M$ bits, with: $b_M \in [10^2,10^{10}]$.

From the top plot of Fig.~\ref{fig:allprotocols-tripartite}, we observe that all four protocols consume a number of preshared bits that scales sub-linearly with respect to the document size ($b_M$). The ERS protocol is the most efficient, closely followed by Protocol~\ref{QDSprot-Chen}, while Protocol~\ref{QDSprot-INDRA} and Protocol~\ref{QDSprot-Andersson} lag behind. For example, to sign a document of $b_M=10^6$ bits, Protocol~\ref{QDSprot-Andersson} requires $\ell_P \approx 1.1 \cdot 10^5$ preshared secret bits and Protocol~\ref{QDSprot-INDRA} uses $\ell_P \approx 10^4$ bits, while Protocol~\ref{QDSprot-Chen} only uses $\ell_P = 441$ bits, marginally beaten by our ERS protocol with $\ell_P=411$ bits. In general, the gap between Protocol~\ref{QDSprot-Chen} and the ERS protocol never exceeds $44$ bits for the considered document sizes.

In the bottom plot of Fig.~\ref{fig:allprotocols-tripartite} we plot the signatures length ($\ell_S$) as a function of the document size. Even in this case, the best performing protocols are Protocol~\ref{QDSprot-Chen} and the ERS protocol, with the latter requiring marginally longer signatures (by about $60$ bits) than the former but never exceeding $200$ bits even for documents of $b_M = 10^{10}$ bits. Conversely, Protocol~\ref{QDSprot-INDRA} and Protocol~\ref{QDSprot-Andersson} require signatures that are $30$ times larger on average.

From our performance analysis in the tripartite scenario, we conclude that our ERS protocol and Protocol~\ref{QDSprot-Chen} are the most efficient protocols in terms of consumed preshared bits, with a consumption up to three orders of magnitude smaller than the other two QDS protocols and that does not exceed $510$ bits for document sizes up to $b_M=10^{10}$ bits. The superior efficiency of the ERS protocol and Protocol~\ref{QDSprot-Chen} is also reflected in the signature lengths, where they beat the other two protocols by more than one order of magnitude.

\section{Field test} \label{sec:experiment}

\begin{table}[h!t]
\renewcommand\arraystretch{1.3}
\setlength{\tabcolsep}{1pt}
\centering
\caption{The standard Poly1305 authenticator as per RFC 8439, where the additional input key $s$ of $128$ bits is employed as a one-time mask of the tag. This promotes the original Poly1305 family from $\varepsilon$-almost $\Delta$ universal$_2$ to $\varepsilon$-almost strongly universal$_2$ (ASU) \cite{Kiktenko_2020}.}
\label{tab:poly1305}
\begin{tabular}[t]{|p{0.25\linewidth}|p{0.35\linewidth}|p{0.25\linewidth}|}
\toprule
\textbf{Family} & \textbf{Preshared bits} & $\varepsilon$ \\
\midrule
$\mathcal{F}_{\rm Poly1305}$ \cite{poly1305} & $256$  & $8 \ceil{\frac{b_M}{16}} / 2^{106}$ \\
\bottomrule
\end{tabular}
\end{table}

In order to validate the ERS protocol in a real-world environment, we developed an application that implements the ERS protocol on the Rome Quantum Metropolitan Area Network (QMAN) \cite{RomeQMAN}. The Rome QMAN is a QKD network currently consisting of 5 nodes dislocated over 4 geographical sites in Rome, Italy, and arranged in a star topology, with the center of the star hosting two independent nodes. The network links span distances from 4.7 km up to 29.7 km. Each node acts as one of the protocol participants and is equipped with a QKD system (either commercial or developed in-house),  a proprietary key management system (KMS) for orchestrating symmetric key delivery between network nodes, and a laptop (Intel Core Ultra 7 Processor 255U, 32GB RAM, 1 Gigabit Ethernet) running the application.

The ERS protocol is implemented  as a digital signature application interfaced with the KMS services offered to  the Application Layer of the QMAN. In each run, the application first requests a sufficient amount of preshared keys from the KMS. This amount depends on the role of the participant (i.e., network node) in that particular run (sender or receiver). Afterwards, the application proceeds as per the description in Sec.~\ref{sec:ERS}.

In this field test, the ERS protocol is implemented in two variants (patent pending). In the ERS-ASU variant, two hashing families $\mathcal{F}_{\rm ASU}$ of the type reported in Table~\ref{tab:efficient-hashing} are adopted for signing documents (with $b_H$-bit hashes) and for establishing the authenticated channels with key recycling (with $b'_H$-bit hashes). The families are constructed according to Ref.~\cite{QDS-Kiktenko} and the hash lengths $b_H$ and $b'_H$ are optimized under the constraint $\varepsilon_{\rm rep} + \varepsilon_{\rm for} \leq 10^{-10}$.

In the ERS-Poly1305 variant, we adopt another $\varepsilon$-almost strongly universal$_2$ family, namely $\mathcal{F}_{\rm Poly1305}$ (c.f. Table~\ref{tab:poly1305}), for both signatures and authenticated channels. In this case, the hash lengths $b_H$ and $b'_H$ are fixed: $b_H=b'_H=128$ bits. A straightforward recalculation of the security parameters against forgery and repudiation of the ERS-Poly1305 variant yields: $\varepsilon_{\rm for}=\max\{2^{-256},8 \ceil{\frac{b_M}{16}} / 2^{106}\}$ and  $\varepsilon_{\rm rep}=8 \ceil{\frac{b_M + 384}{16}} / 2^{106}$. Therefore, for the document sizes chosen in the field test (up to $b_M=10^9$ bits), we have: $\varepsilon_{\rm for} + \varepsilon_{\rm rep} \approx 1.2 \cdot 10^{-23}$, which is well below our security threshold of $10^{-10}$. 

We tested both ERS protocol variants for signing documents of $b_M=10^r$ bits, for $r=2,3,4,5,6,7,8,9$, which are then verified by four distinct receivers ($N=4$). In the top plot of Fig.~\ref{fig:fieldtest-time}, we plot the total runtime of the two protocol variants against the document size. Specifically, we measured the time needed to run the whole ERS protocol from signature generation until the successful verification of the signature by all receivers. Note that this time interval does not include the time required to generate the pairwise keys consumed by the protocol. The reason for this is that, typically, QKD keys are constantly generated throughout the network and are readily available when an application like the ERS protocol demands them, thus introducing no time delay. From the plot, we see that the ERS-Poly1305 variant is much faster than the ERS-ASU version, with a sub-second runtime up to $10$-Mbit documents, due to a faster computation of the hashes. In particular, the hash functions of the $\mathcal{F}_{\rm ASU}$ family require the creation of, and operations on, finite fields, which is computationally demanding.

The bottom plot of Fig.~\ref{fig:fieldtest-time} shows the number of preshared bits consumed by each receiver in one protocol run. Here, both ERS variants reduce key consumption by more than two orders of magnitude compared to Protocol~\ref{QDSprot-Andersson} from Ref.~\cite{USS-andersson} (compare with bottom plot in Fig.~\ref{fig:preshared_signature-ERS}), with the ERS-ASU variant being more efficient due to the optimization over hash lengths, while the hashes of the $\mathcal{F}_{\rm Poly1305}$ family are fixed. Note that this level of key consumption is easily sustainable with modern QKD networks, where available key pools are typically much larger than few thousands bits and key generation rates are sufficient to sign several large documents every minute.

In summary, our field test shows that ERS-Poly1305 can generate up to 13 signatures per second on one-Mbit documents across a deployed metropolitan network, while consuming less than 5000 preshared bits per receiver. Thus, the ERS protocol emerges as a strong candidate for providing digital signature services in modern QKD networks.

\begin{figure}[htb]
    \centering
    \includegraphics[width=1\linewidth,keepaspectratio]{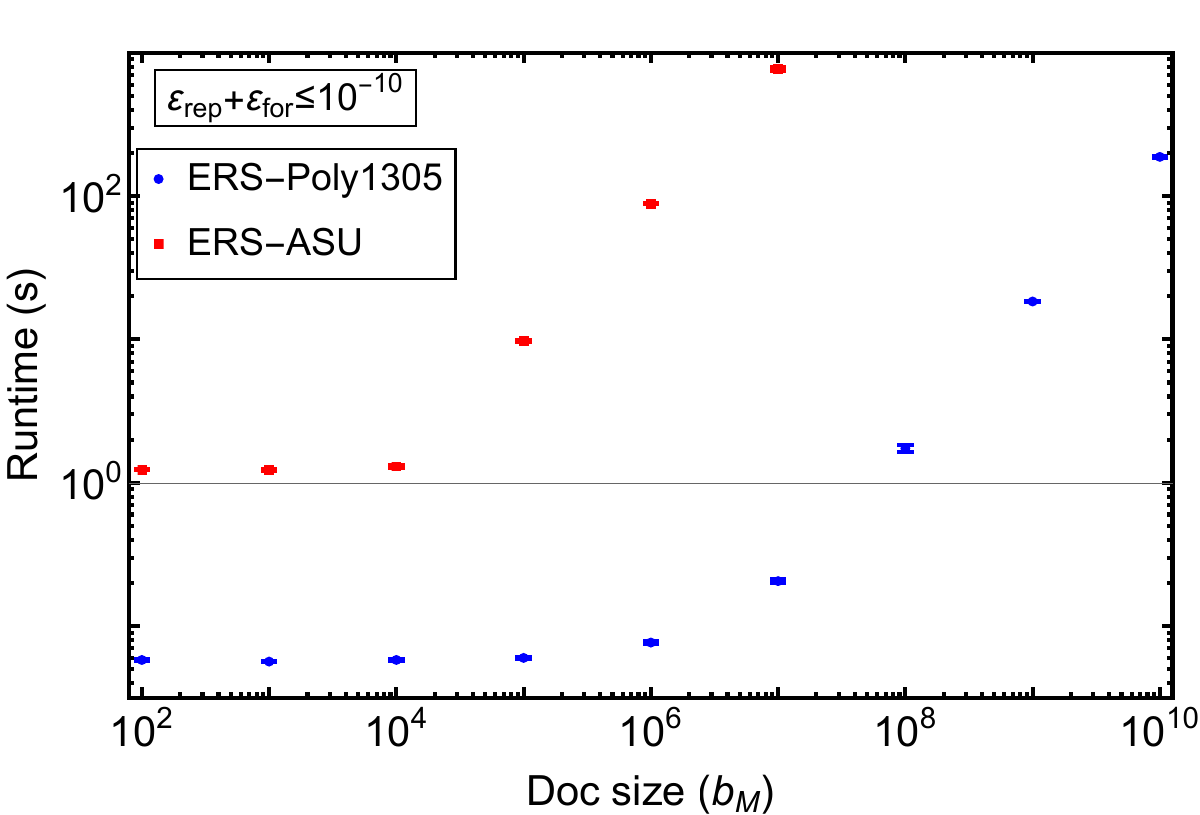}
    \includegraphics[width=1\linewidth,keepaspectratio]{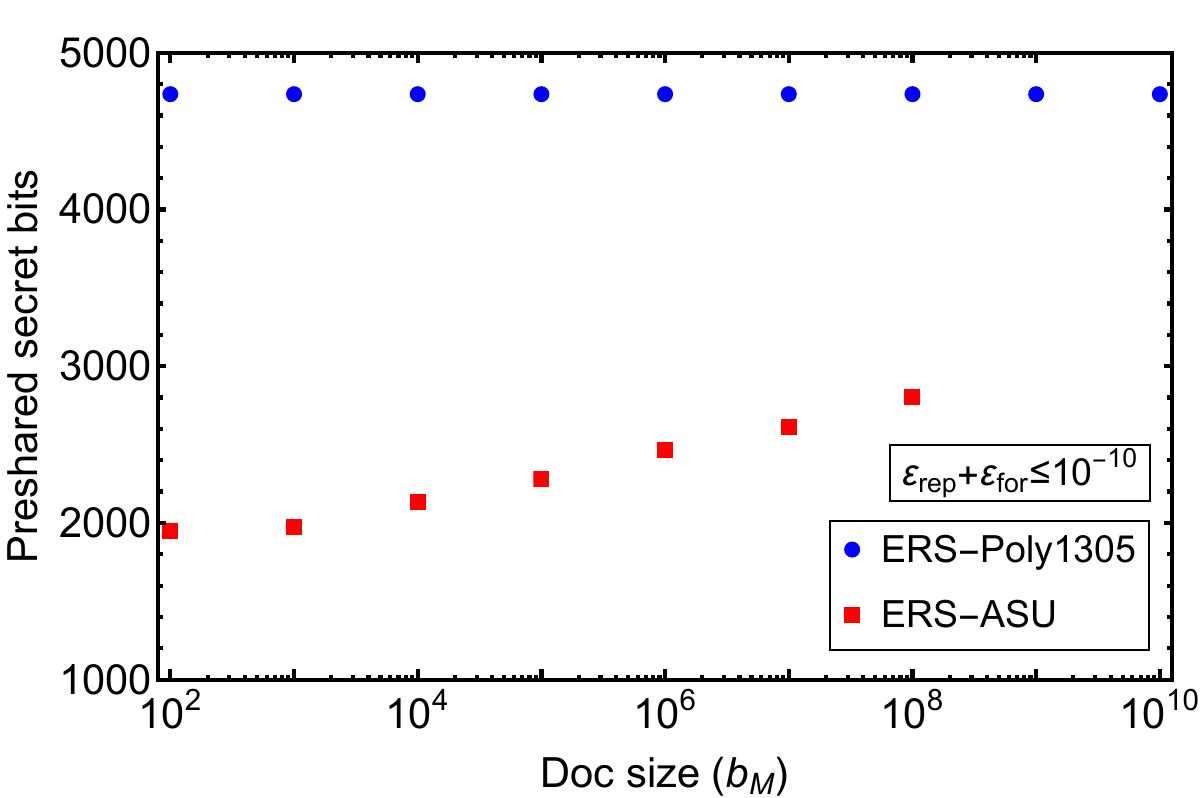}
    \caption{The total runtime of the application implementing the two variants of the ERS protocol (top) and the consumed preshared secret bits per receiver (bottom), as a function of the document size ($b_M$) and for security parameters satisfying $\varepsilon_{\rm rep} + \varepsilon_{\rm for} \leq 10^{-10}$.  In the top plot, each point represents the average runtime over four samples and the vertical error bars correspond to $\pm 1$ standard deviation.}
    \label{fig:fieldtest-time}
\end{figure}

\section{Conclusion} \label{sec:conclusion}

In the first part of this work, we critically reviewed three practical quantum digital signature (QDS) protocols \cite{QDS-Chen,QDS-INDRA,USS-andersson}, focusing on their security claims and potential vulnerabilities. The three QDS protocols were selected for their ability to sign large documents with relatively small signatures and without a trusted authority, by relying on previously-established secret keys (e.g. through quantum key distribution).

Our analysis highlighted that, in the multiparty scenario, the QDS protocol of Ref.~\cite{USS-andersson} -- and many other multiparty QDS proposals \cite{multiQDS-Chen,multiQDS-Korolkova,potentialmultiQDS-Toshiba,multiQDS-Ma} -- can be vulnerable to equivocation attacks, whereby signature repudiations and transferability failures can be enforced by corrupting the document-signature pair analyzed by different honest receivers, without them realizing.

We addressed this vulnerability by developing an equivocation-resistant signature (ERS) protocol based on preshared keys (patent pending). The design of the ERS protocol deliberately departs from established multiparty QDS paradigms \cite{USS-andersson,multiQDS-Chen,multiQDS-Korolkova,potentialmultiQDS-Toshiba,multiQDS-Ma}, in order to overcome the identified vulnerability and gain efficiency in terms of consumed preshared bits and signature length. Indeed, in the $N$-receiver scenario, the ERS protocol drastically reduces signature sizes by at least a factor $N^2$ compared to the protocol of Ref.~\cite{USS-andersson} and it reduces preshared key consumption by more than two orders of magnitude independently of $N$ and of the document size. The superior performance of the ERS protocol is confirmed when benchmarked against the three promising QDS protocols in the tripartite scenario.

Our numerical results corroborate the different design choice of the ERS protocol compared to the more common approach adopted by Protocol~\ref{QDSprot-INDRA} and Protocol~\ref{QDSprot-Andersson} (c.f. Subsec.~\ref{sec:security-discussion}). In detail, in the latter protocols every receiver is capable of independently verifying the signature as soon as the sender delivers the message, but this is extremely inefficient in terms of preshared key consumption and signature length. Conversely, our ERS protocol ensures global receiver consensus on the signature before enabling its independent verification, leading to large reductions in resource consumption without compromising on speed. Indeed, in the field trial the ERS protocol reached a signature rate of 13 tps for one-megabit documents shared among five users.

In conclusion, our work provides an information-theoretic secure and practical digital signature protocol, which exclusively relies on pairwise secret keys provided, for example, by QKD. The ERS protocol is immune to equivocation-based attacks while resulting overwhelmingly more efficient than previous multiparty QDS proposals. The combination of improved efficiency and security places our ERS protocol as a primal candidate for implementation in today's quantum key distribution networks, thereby enhancing their cryptographic services beyond mere symmetric key establishment.

\bibliography{bibliography}

\clearpage
\onecolumngrid

\appendix

\section{Wegman-Carter authentication} \label{app:WC}

The two universal hashing families summarized in Table~\ref{tab:efficient-hashing} can be used for establishing authenticated classical channels with information-theoretic (IT) security through the Wegman-Carter (WC) authentication method \cite{WC}.

In brief, suppose that a sender and a receiver share a random secret string that uniquely identifies a hash function $f$ from $\mathcal{F}_{\rm ASU}$. Then, the sender wishing to send a message $m$ sends the tuple $(m,t)$ to the receiver over a public channel, where the tag is $t=f(m)$. Let $(m',t')$ be the tuple received the by the receiver. The receiver accepts the message as authentic if $f(m')=t'$. By combining the properties \eqref{eq-almost-strongly-universal1} and \eqref{eq-almost-strongly-universal2} of $\varepsilon$-ASU$_2$ families, we deduce that an attacker trying to forge a pair $(m',t')$ with $m \neq m' $ that is accepted by Bob only succeeds with probability $\varepsilon$ \cite{Kiktenko_2020}.

An analogous result can be proved for $f \in \mathcal{F}_{\rm AXU}$. Indeed, if the sender and receiver additionally share a random string $r$ and the tag is computed as $t=f(m) \oplus r$ (the symbol $\oplus$ indicates addition modulo $2$), the property \eqref{eq-almost-XOR-universal2} of $\varepsilon$-AXU$_2$ families ensures that an attacker forging a new pair only succeeds with probability $\varepsilon$ \cite{LFSR-hashing}.

Interestingly, as first proposed by WC \cite{WC}, it is possible to authenticate $n$ messages using the same hash function from either family, $\mathcal{F}_{\rm AXU}$ or $\mathcal{F}_{\rm ASU}$, by OTP-encrypting the tags with fresh random strings, such that the attack probability on any of the $n$ messages is still $\varepsilon$ \cite{Kiktenko_2020}. This procedure is sometimes called key recycling. In the following, we first illustrate the WC method with key recycling and then prove its IT security when using the $\varepsilon$-ASU$_2$ family.\\

\begin{protocol}  \label{WCprot} \caption{WC authentication \cite{WC}}
\begin{enumerate}[wide, labelwidth=!, labelindent=0pt]
\item Let $M$ be the set of possible messages and $B$ the set of hashes, or tags. Let $\mathcal{F}$ be a publicly-known $\varepsilon$-ASU$_2$ family of hash functions from $M$ to $B$. 

\item Alice and Bob agree on the total number of messages $n$ they want to authenticate and they agree on a uniformly random secret key $(k_0, (k_1, \dots, k_n))$. The subkey $k_0$ identifies a unique hash function $f_{k_0} \in \mathcal{F}$ used to generate the tags of each of the $n$ authenticated messages.

\item For $i \in \{1,2,\dots,n\}$, do the following:
\begin{enumerate}
    \item Alice chooses a message $m_i \in M$ to be sent and generates the corresponding tag $t_i$ as: $t_i=f_{k_0} (m_i) \oplus k_i$. She sends the message $m_i$, the tag $t_i$ and the index $i$ to Bob.

    \item Bob uses the subkey $k_i$, corresponding to the received index from Alice, to extract the hash from the received tag by computing: $h_i = t_i \oplus k_i$. Bob then compares $h_i$ with the hash resulting from applying the selected hash function on the received message: $f_{k_0}(m_i)$. If they match, Bob authenticates Alice's message $m_i$.
\end{enumerate} 
\end{enumerate}
\end{protocol}

\vspace{1em}

Here we prove the IT security of the WC authentication method implemented with an $\varepsilon$-ASU$_2$ family.

\begin{theorem} \label{th:WC-security}
    Let $(k_0, (k_1, \dots, k_n))$ be a randomly chosen key known to Alice and Bob and let $m_1,\dots,m_n$ be $n$ messages sent by Alice to Bob via the WC authentication method. Suppose the attacker, Eve, knows the $\varepsilon$-ASU$_2$ family $\mathcal{F}$, the set of messages $M$, the messages $m_1,\dots,m_n$, their tags $t_1,\dots, t_n$ and indexes. Then, there is no forged message $m_E$ (for any index $i$) for which Eve can guess the correct tag $t_E$ with a probability larger than $\varepsilon$.
\end{theorem}

\begin{proof}
    Suppose, without loss of generality, that Eve wants to replace the first message $m_1$ with a forged message $m_E$. Then, the correct tag such that Bob authenticates the forged message would be: $t_E=f_{k_0}(m_E) \oplus k_1$. Note that Eve does not know the preshared key $(k_0, (k_1, \dots, k_n))$, hence she needs a strategy to correctly guess $t_E$. An optimal strategy consists in fixing a value for the tag, $t$, and consider all the possible keys $(k_0, (k_1, \dots, k_n))$ for which the observed message-tag pairs $(m_i,t_i)$ and the forged pair $(m_E,t)$ are successfully authenticated. Let us denote this set $S(t)$:
    \begin{align}
        S(t) = \left\lbrace (k_0, (k_1, \dots, k_n)) : \, t=f_{k_0}(m_E) \oplus k_1, \, t_1=f_{k_0}(m_1)\oplus k_1, \, t_i=f_{k_0}(m_i)\oplus k_i \, \mbox{for } i=2,\dots,n \right\rbrace.
    \end{align}
    Then, Eve's guess of the tag corresponds to the value for which there are most key combinations allowed, i.e., 
    \begin{align}
        t_{\rm guess} := \arg \max_{t \in B} \abs{S(t)},
    \end{align}
    where $\abs{S}$ is the cardinality of $S$. Indeed, this guess maximizes Eve's probability of guessing the correct tag. In particular, given messages $m_1,\dots,m_n$, tags $t_1,\dots, t_n$ and a forged message $m_E$, Eve's guess is deterministic and given by $t_{\rm guess}$. The probability that $ t_{\rm guess} $ is correct corresponds to the probability that the key $(k_0, (k_1, \dots, k_n))$, randomly chosen by Alice and Bob, is such that $t_E=f_{k_0}(m_E) \oplus k_1 = t_{\rm guess}$, conditioned on the compatibility with the observed message-tag pairs. Then, we have that Eve's guessing probability is given by:
    \begin{align}
        p_{\rm guess} = \Pr_{(k_0, (k_1, \dots, k_n))}[t_E = t_{\rm guess}|t_i=f_{k_0}(m_i)\oplus k_i \, \mbox{for } i=1,\dots,n]
        &= \frac{\abs{S(t_{\rm guess})}}{\abs{\mathcal{F}}}, \label{Eveguessprob}
    \end{align}
    where we used the fact that the keys are chosen randomly by Alice and Bob and that their total number (given that  $(m_i,t_i)$ are authenticated) is
    \begin{align}
        \abs{\left\lbrace (k_0, (k_1, \dots, k_n)) : \, t_i=f_{k_0}(m_i)\oplus k_i \, \mbox{for } i=1,\dots,n \right\rbrace} = \abs{\mathcal{F}}.
    \end{align}
    As a matter of fact, once we fix $k_0$, then all other keys are fixed by the equations $k_i = t_i \oplus f_{k_0} (m_i)$.
    
    In the following, we compute Eve's guessing probability starting from the cardinality of $S(t)$, $\abs{S(t)}$. For a fixed $k_1$, there are at most $\varepsilon\abs{\mathcal{F}}/\abs{B}$ possible choices of $f \in\mathcal{F}$ such that $f(m_E)=t \oplus k_1 \wedge f(m_1)=t_1 \oplus k_1$ due to the property \eqref{eq-almost-strongly-universal2} of an $\varepsilon$-ASU$_2$ family. Therefore, there are at most $\varepsilon\abs{\mathcal{F}}/\abs{B}$ possible values of $k_0$ for each given $k_1$. Now, for each given value of $k_0$, there is a unique string of keys $k_2,\dots,k_n$ such that the remaining conditions of the set $S(t)$ are satisfied; these keys are fixed by: $k_i = t_i \oplus f_{k_0}(m_i)$. Thus, summing up, there are at most  $\abs{B} \cdot \varepsilon \abs{\mathcal{F}}/\abs{B}$ keys in $S(t)$:
    \begin{align}
        \abs{S(t)} \leq  \varepsilon \abs{\mathcal{F}} = \abs{S(t_{\rm guess})}, \label{proof-eq-1}
    \end{align}
    where we used that $t=t_{\rm guess}$ maximizes the size of $S(t)$.
    By employing \eqref{proof-eq-1} in \eqref{Eveguessprob}, we obtain Eve's guessing probability
    \begin{align}
        p_{\rm guess} = \varepsilon,
    \end{align}
    which concludes the proof.
\end{proof}

\section{Efficient hashing families} \label{app:efficient-hashing}

In this Appendix we specify the $\varepsilon$-ASU$_2$ family and the $\varepsilon$-AXU$_2$ family adopted in the analyzed QDS protocols. We follow the notation laid out in Sec.~\ref{sec:notation} such that each hash function maps strings of $b_M$ bits to strings of $b_H$ bits.

The chosen families are selected for their efficiency in the number of preshared key bits required to uniquely identify a function of the family. In particular, the number of preshared bits scales with $\log_2 b_M$ (for a fixed security parameter $\varepsilon$), rather than with $b_M$ as for strongly-universal$_2$ sets based on random matrices.\\

\textbf{$\varepsilon$-ASU$_2$ family:} \quad The $\varepsilon$-ASU$_2$ family employed in this manuscript is denoted $\mathcal{F}_{\rm ASU}$, its security parameter is $\varepsilon=2^{1-b_H}$ and its explicit construction is provided in Refs.~\cite{ASU-Andersson,QDS-Kiktenko}. Each function of the set $\mathcal{F}_{\rm ASU}$ is uniquely defined by a string of $y$ bits, where $y$ satisfies: $y=3 b_H + 2 \sigma$, where $\sigma$ is the smallest number that satisfies: $b_M \leq (b_H + \sigma)(1+ 2^\sigma)$. Since $\sigma$ is defined by a transcendental inequality, it can only be computed numerically. Alternatively, it might be convenient to obtain a relatively tight analytical upper bound on $\sigma$ by solving: $b_M = b_H (1+ 2^\sigma)$, which yields
\begin{align}
    \bar{\sigma} = \log_2 \left(\frac{b_M}{b_H} -1 \right).
\end{align}
This provides us with the following upper bound on the number of bits $y$ required to specify an element in $\mathcal{F}_{\rm ASU}$,
\begin{align}
    \bar{y} = \left\lceil 3 b_H + 2  \log_2 \left(\frac{b_M}{b_H} -1\right) \right\rceil,
\end{align}
and this is the value that we use in our performance analysis of Sec.~\ref{sec:performance}.\\

\textbf{$\varepsilon$-AXU$_2$ family:} \quad The $\varepsilon$-AXU$_2$ family employed in this manuscript is denoted $\mathcal{F}_{\rm AXU}$, with security parameter $\varepsilon=b_M 2^{1-b_H}$, and it is composed of Toeplitz matrices such that consecutive columns are consecutive states of a linear feedback shift register (LFSR) of length $b_H$ . The family was originally introduced in Ref.~\cite{LFSR-hashing}. Therefore, each hash function, i.e. each Toeplitz matrix, is specified by an LFSR and its initial state, totaling to $2b_H$ bits.

In order to specify the elements of $\mathcal{F}_{\rm AXU}$, we first define LFSRs.

\begin{definition}
    Let $p(x)$ be a polynomial of degree $n$ over GF(2), $p(x)=x^n + p_{n-1} x^{n-1} + \dots p_1 x + p_0$ and $s^0=(s_{n-1},\dots,s_1,s_0)$ a binary string. A linear feedback shift register (LFSR) of length $n$ is specified by $p(x)$ and the initial state of the register, $s^0$. The following state of the register, $s^1$, is obtained by shifting to the right the bits of the previous state, $s^0$, and by computing the new element: $s_n= s^0 \cdot p \mod 2$, where $p$ is the vector of coefficients $p=(p_{n-1},\dots,p_0)$ and where ''$\cdot$'' indicates the scalar product. Thus we have: $s^1=(s_n,s_{n-1},\dots,s_1)$. By applying the same rule recursively, the $k$-th state of the LFSR is given by: $s^k=(s_{n-1+k},\dots,s_{1+k},s_k)$, where a generic element of the sequence of bits generated by the LFSR is: $s_l=s^{l-n} \cdot p \mod 2$, for $l \geq n$.
\end{definition}

From the above definition, we deduce that the transpose of consecutive states of an LFSR can be considered as consecutive columns of a Toeplitz matrix. We now define the elements of $\mathcal{F}_{\rm AXU}$.

\begin{definition}
    Let $p(x)$ be an irreducible polynomial of degree $b_H$ over GF(2) and let $s^0=(s_{b_H -1}, \dots, s_1,s_0)^T$ be the initial state of the LFSR with connection polynomial $p(x)$, with $s^0 \neq 0$. Let $T_{p,s}$ be the Toeplitz matrix with columns given by: $T_{p,s}=(s^0,s^1,\dots,s^{b_M-1})$, where $s^k$ is the transpose of the $k$-th state of the LFSR. The hash function $f_{p,s} \in \mathcal{F}_{\rm AXU}$ maps messages $m=(m_0,\dots,m_{b_M-1})^T$ of variable length up to $b_M$ bits to $b_H$-bit strings given by: $f_{p,s}(m)= T_{p,s} m \mod 2$. In other words, the $j$-th element of the hash reads: $(f_{p,s}(m))_j=\bigoplus_{i=0}^{b_M-1} m_i \, s_{b_H-j+i}$, where $s_{b_H-j+i}$ is the $(b_H-j+i)$-th element of the LFSR sequence. 
\end{definition}

Note that to avoid allowing an adversary to add zeroes to the message without changing the tag, each message needs to end with a $1$.

\section{Security proof of Protocol~\ref{QDSprot-Chen}} \label{app:security-Chen}

\begin{lemma*}
    Protocol~\ref{QDSprot-Chen} is $\varepsilon_{\rm for}$-secure against forgery according to Definition~\ref{def:forgery}, with $\varepsilon_{\rm for}=b_M/2^{b_H-1}$.
\end{lemma*}

\begin{proof}
    There are two cases of forgery attacks: 1) Alice does not sign any document and Bob forges a new pair $\{Doc',Sig'\}$; 2) Alice sends $\{Doc,Sig\}$ to Bob and Bob forges a modified pair $\{Doc',Sig'\}$.

    In the first case, Bob has no information from Alice and in particular no information on Alice's key $X_A^{2b_H}$. Suppose the worst-case scenario where Bob forwards to Charlie a document $Doc'$ and a digest $Dig'=(h_a',p_a')$ that is consistent with the document: $h_a'=T_{p_a',X^{b_H}_A} \cdot Doc'$. Charlie will recover the digest $Dig'$ and hence validate the signature if and only if Bob correctly guesses the encryption key $X^{2b_H}_A$ used to encrypt the digest. Therefore, the attack will succeed with probability $2^{-2b_H}$.

    In the second case, Bob's goal is to find a message $m$ and a hash $t$ such that: $t=T_{p_a,X_A^{b_H}} \cdot m$. Indeed, if he succeeds, he sends $\{Doc',Sig'\}$ with $Doc'=Doc \oplus m$ and $Sig'=Sig \oplus (t || 0)$. This pair will be validated by Charlie, since the expected hash by Charlie is: $h_c=T_{p_a,X_A^{b_H}} \cdot (Doc \oplus m) = h_a \oplus t$, which coincides with the hash sent by Bob in $Sig'$. Therefore, the goal of Bob coincides with the goal of an adversary in WC authentication implemented with the family $\mathcal{F}_{\rm AXU}$, where the tags are protected by OTP. From Table~\ref{tab:efficient-hashing}, we have that Bob's attack succeeds at most with probability $b_M/2^{b_H-1}$, which is larger than the success probability of the attack in the first analyzed case.

    In conclusion, by putting together the two cases, Bob's forgery attacks succeed with probability at most $\varepsilon_{\rm for}=b_M/2^{b_H-1}$, which concludes the proof.    
\end{proof}

\begin{lemma*}
    Protocol~\ref{QDSprot-Chen}, with $b_M>b_H$, is $\varepsilon_{\rm rep}$-secure against repudiation according to Definition~\ref{def:repudiation}, with $\varepsilon_{\rm rep}=(2b_H+b_M)/2^{b'_H-1}$.
\end{lemma*}

\begin{proof}
    Since Bob and Charlie communicate via an authenticated channel and behave honestly, Charlie obtains the same pair $\{Doc,Sig\}$ received by Bob and the key $X_B$. Bob in turn obtains $X_C$. Thus, both Bob and Charlie recover Alice's key: $K_B=K_C=X_A$. Using $X_A$ and the pair $\{Doc,Sig\}$, both Bob and Charlie will reach the same conclusion about Alice's signature. Therefore, it is impossible that they disagree, unless Alice successfully tampers with the authenticated channel.
    
    Given the implementation of the authenticated channel described in Sec.~\ref{sec:notation}, the probability that Alice modifies Bob's message $\{Doc,Sig\}$, or the message $X_B$, and that the message is authenticated by Charlie is at most $(2b_H + b_M)/2^{b'_H-1}$ ($3b_H/2^{b'_H-1}$ for the message $X_B$). Similarly, the probability of a successful attack on the message $X_C$ from Charlie is $3b_H/2^{b'_H-1}$. Since the protocol aborts when any message is not authenticated, the maximum probability that the protocol does not abort and Alice successfully modifies an authenticated message --thereby causing a repudiation-- is given by: $\varepsilon_{\rm rep}=\max\{(2b_H+b_M)/2^{b'_H-1},3b_H/2^{b'_H-1}\}$, which concludes the proof.
\end{proof}

\begin{lemma*}
    If $p_a$ is known to Bob, then he can always perform a successful forgery attack, i.e., Charlie validates the pair $\{Doc',Sig'\}$ forged by Bob with unit probability.
\end{lemma*}

\begin{proof}
To see this, we observe that Bob's goal is to find a $b_M$-bit message $m$ and a hash $t$ such that: $t=T_{p_a,X_A^{b_H}} \cdot m$. If he succeeds, he can send $\{Doc',Sig'\}$ with $Doc'=Doc \oplus m$ and $Sig'=Sig \oplus (t || 0)$, which will be validated by Charlie, since the expected hash by Charlie is: $h_c=T_{p_a,X_A^{b_H}} \cdot (Doc \oplus m) = h_a \oplus t$, which coincides with the hash sent by Bob in $Sig'$. Now, from the knowledge of $p_a$, it is easy for Bob to find a valid hash $t$ and message $m$ that satisfy $t=T_{p_a,X_A^{b_H}} \cdot m$. For instance, Bob can choose $t=0$ (the null vector). From Theorem~8 in Ref.~\cite{LFSR-hashing}, we have that:
\begin{align}
    t &= T_{p_a,X_A^{b_H}} \cdot m \nonumber\\
    &=B D B^{-1} \cdot X_A^{b_H},
\end{align}
where $B$ is a non-singular $b_H \times b_H$ matrix that only depends on $p_a$ and $D$ is a diagonal matrix with elements $m(\lambda_i)$, where $m(x)$ is the polynomial associated to the message $m$ and $\lambda_1,\lambda_2,\dots,\lambda_{b_H}$ are the $b_H$ roots of the irreducible polynomial $p_a(x)$ over GF($2^{b_H}$). Having fixed $t=0$, it is enough for Bob to find a polynomial $m(x)$ such that $BDB^{-1}$ is the null matrix, i.e. such that $m(\lambda_i)=0$ for every $i$. In other words, Bob wants that every root of $p_a(x)$ (over GF($2^{b_H}$)) is also a root of $m(x)$. This can be readily achieved by choosing a message $m$ such that $p_a(x)$ divides $m(x)$, i.e. a polynomial that can be decomposed as: $m(x)=p_a(x) \cdot g(x)$, where $g(x)$ is an arbitrary polynomial over GF(2) of degree $b_M-b_H$. Then, the pair $t=0$ and the message associated to the polynomial $m(x)=p_a(x) \cdot g(x)$ are such that  $t=T_{p_a,X_A^{b_H}} \cdot m$ holds, which concludes the proof.  
\end{proof}

\section{Security proof of Protocol~\ref{QDSprot-INDRA}} \label{app:security-INDRA}

\begin{lemma*}
    Protocol~\ref{QDSprot-INDRA} is $\varepsilon_{\rm for}$-secure against forgery according to Definition~\ref{def:forgery}, with
    \begin{align}
        \varepsilon_{\rm for} = 
        \Xi(n/2-e_{max},n/2,b_M 2^{1-b_H}) ,
    \end{align}
    where the function $\Xi(k,n,p)$, defined as
    \begin{align}
        \Xi(k,n,p):= \sum_{j=k}^{n} \binom{n}{j} p^j (1-p)^{n-j},
    \end{align}
    represents the probability of obtaining at least $k$ successes in $n$ trials with success probability $p$ for each trial.
\end{lemma*}

\begin{proof}
    There are two cases of forgery attacks: 1) Alice does not sign any document and Bob forges a new pair $\{Doc',Sig'\}$; 2) Alice sends $\{Doc,Sig\}$ to Bob and Bob forges a modified pair $\{Doc',Sig'\}$.

    Let us first discuss case 1). In this case, Bob generates a document $Doc'$ and needs to provide Charlie with the correct signature. In particular, since Charlie only verifies a subset of signatures, given by: $Sig'_{n+1}, \dots, Sig'_{2n}$ and $Sig'_{\gamma_B(1)}, \dots, Sig'_{\gamma_B(n/2)}$, Bob needs to provide valid signatures for this subset. For the signatures $Sig'_{\gamma_B(1)}, \dots, Sig'_{\gamma_B(n/2)}$, Bob knows the corresponding decryption keys from $X_B$, in particular the strings $r^{\gamma_B(i)}_B$, which he passed to Charlie in the distribution stage. Thus, Charlie will see no error when verifying these signatures.
    
    In contrast, for the $n$ signatures $Sig'_{n+1}, \dots, Sig'_{2n}$, Bob only knows $n/2$ decryption keys, i.e. the keys $X^{\gamma_C(i)}_C$ (for $i=1,\dots,n/2$) that he received from Charlie. In particular, Bob does not know the strings $r^j_C$ for the $n/2$ keys that he did not receive from Charlie and needs to guess them in order for Charlie to approve his forged signatures. However, we recall that Charlie is allowed up to $e_{max}$ errors during verification. Therefore, Bob only needs to correctly guess the string $r^j_C$ at least $n/2-e_{max}$ times out of $n/2$ possibilities. Since each string $r^j_C$ is $2b_H$ bits long, the probability for $k\geq n/2-e_{max}$ successes out of $n/2$ trials reads:
    \begin{align}
        \Xi(n/2-e_{max},n/2,2^{-2b_H}). \label{INDRA-forgery-1}
    \end{align}
    Thus, the quantity in \eqref{INDRA-forgery-1} represents the probability of a successful forgery attack by Bob for case 1).
    
    We now discuss case 2). Here, Bob receives a valid pair $\{Doc,Sig\}$ from Alice and wants to forge a new pair  $\{Doc',Sig'\}$ where $Doc' \neq Doc$. Like before, Bob can produce new valid signatures $Sig'_{\gamma_B(1)}, \dots, Sig'_{\gamma_B(n/2)}$ since Bob holds the corresponding keys and overrides Alice's choices of irreducible polynomials with his choices of polynomials. Likewise, Bob can forge the signatures $Sig'_{n+\gamma_C(i)}$ corresponding to the keys $X^{\gamma_C(i)}_C$ (for $i=1,\dots,n/2$) that he received from Charlie.
    
    For the remaining $n/2$ signatures $Sig_{j+n}$ received from Alice for which Bob does not hold the corresponding key, Bob has two possibilities.
    
    He can act as in case 1), by neglecting the signatures sent by Alice and forging new signatures. However, this strategy requires Bob to guess some of the encryption keys $r^j_C$. The probability of correctly guessing one of these keys is $2^{-2b_H}$.
    
    The second possibility is also Bob's best strategy since it succeeds with a higher probability: Bob wants to guess a message-tag pair $(m_j,t_j)$, with $m_j\neq 0$, such that $t_j=T_{p^j_C,s^j_C}\cdot m_j$ is satisfied. In doing so, note that Bob does not know either parameter of the Toeplitz matrix $T_{p^j_C,s^j_C}$ and that the tag in $Sig_{j+n}$ is protected by OTP. If Bob succeeds, he can forge the signature $Sig'_{j+n}=Sig_{j+n} \oplus (t_j||0)$ and choose the document $Doc'=Doc \oplus m_j$ such that Charlie will validate the forged signature with unit probability. The probability that Bob finds the correct pair $(m_j,t_j)$ is the probability of a successful attack on WG authentication implemented with the $\mathcal{F}_{\rm AXU}$ family and reads $b_M/2^{b_H-1}$ (see Table~\ref{tab:efficient-hashing}), where $b_M$ is the length of the message. Let us assume that Bob chooses the same message $m_1=m_2=\dots=m$ for each of the $n/2$ signatures, such that the forged document is uniquely defined by $Doc'=Doc \oplus m$. Then, the probability that Bob correctly guesses at least $n/2 - e_{max}$ pairs $(m,t_j)$, out of $n/2$ signatures, is given by:
    \begin{align}
        \Xi(n/2-e_{max},n/2,b_M 2^{1-b_H}), \label{INDRA-forgery-2}
    \end{align}
    and represents the probability of a successful forgery attack in the case 2).

    By comparing the probability of successful attacks from cases 1), Eq.~\eqref{INDRA-forgery-1} and 2), Eq.~\eqref{INDRA-forgery-2}, we deduce that the attack from case 2) is always more likely to occur since $b_M 2^{1-b_H} > 2^{-2 b_H}$. This concludes the proof.
\end{proof}

\begin{lemma*}
    Protocol~\ref{QDSprot-INDRA}, with $b_M + 4 n b_H > (n/2)\log_2 n$, is $\varepsilon_{\rm rep}$-secure against repudiation according to Definition~\ref{def:repudiation}, with
    \begin{align}
        \varepsilon_{\rm rep}=\max\left\lbrace \prod_{i=0}^{e_{max}} \frac{n/2 - i}{n-i}, \frac{b_M+4nb_H}{2^{b'_H-1}} \right\rbrace.
    \end{align}
\end{lemma*}

\begin{proof}
    Let us first assume that the the authenticated channel between Bob and Charlie is ideal. In this case, Alice's only chance to force a rejection at Charlie is to deliberately introduce errors in some of the signatures verified by Charlie, namely $Sig_{n+1}, \dots Sig_{2n}$, such that Charlie observes $e_{max}+1$ errors in verification, thus rejecting the signature. However, Bob verifies $n/2$ of these signatures, namely $Sig_{n+\gamma_C(1)}, \dots, Sig_{n+\gamma_C(n/2)}$, through the keys $X_C^{\gamma_C(i)}$ (for $i=1,\dots,n/2$) received from Charlie. Since Bob accepts no error in verification and since the $n/2$ key blocks sent by Charlie to Bob are randomly selected, the probability that Alice distributes the $e_{max}+1$ errors in the $n/2$ signatures that are not verified by Bob is given by:
    \begin{align}
        \varepsilon &= \frac{\binom{n/2}{e_{max}+1}}{\binom{n}{e_{max}+1}} \nonumber\\
        &= \prod_{i=0}^{e_{max}} \frac{n/2 - i}{n-i}. \label{INDRA-repudiation1}
    \end{align}
    Note that if only one error is introduced by Alice in the signatures verified by Bob, Bob will reject the signature and the protocol aborts.
    
    Let us now consider the failure of the authenticated channel between Bob and Charlie, which we recall is implemented with an $\varepsilon$-AXU$_2$ family based on LFSRs with hashes of $b_H'$ bits. In order to force a repudiation, Alice can attempt to modify enough key blocks sent to Charlie by Bob. The probability that Alice successfully modifies the key blocks from $X'_B$ sent by Bob through the authenticated channel is at most $(3b_H n/2)/2^{b'_H-1}$. Alternatively, Alice can modify the positions $\gamma_B(1), \dots, \gamma_B(n/2)$ sent by Bob, such that Charlie will hold the right blocks but in the wrong positions. In this case, the probability of a successful attack is $[(n/2)\log_2 n]/2^{b'_H-1}$. Finally, Alice could attempt to modify the pair $\{Doc,Sig\}$ when forwarded by Bob, with a success probability of at most $(b_M+4nb_H)/2^{b'_H-1}$. Since the protocol aborts when a message from Bob is not authenticated by Charlie, the maximum probability of a successful repudiation attack based on the failure of authentication is:
    \begin{align}
        \varepsilon' &= \frac{\max\left\lbrace3b_H n/2, (n/2)\log_2 n, b_M+4nb_H  \right\rbrace}{2^{b'_H-1}} \nonumber\\
        &=(b_M+4nb_H)/2^{b'_H-1}. \label{INDRA-repudiation2}
    \end{align}

    By combining \eqref{INDRA-repudiation1} and \eqref{INDRA-repudiation2}, the maximal probability that the protocol does not abort and that Charlie rejects the signature while Bob accepts it is $\max\{\varepsilon, \varepsilon'\}$, which concludes the proof.  
    
\end{proof}

\section{Security proof of Protocol~\ref{QDSprot-Andersson}} \label{app:security-Andersson}

In this Appendix we prove the IT security of the QDS protocol presented in Sec.~\ref{sec:QDS-Andersson}.

\begin{lemma*} 
    Protocol~\ref{QDSprot-Andersson} is $\varepsilon_{\rm transf}$-secure against non-transferability according to Definition~\ref{def:transferability-Andersson}, with:
    \begin{align}
        \varepsilon_{\rm transf}= \binom{N(1-d_R)}{2} \max\{\varepsilon_{i,j},\varepsilon_{\rm auth},\varepsilon_{\rm hyb}\},
    \end{align}
    where
    \begin{align}
        \varepsilon_{i,j} &= N(1-d_R) \exp\left(- \frac{k}{8(l_{max}+1)^2}\right) \\
        \varepsilon_{\rm auth} &=\left(\frac{ky + k\log_2(Nk)}{2^{b'_H-1}}\right)^{\ceil{N/2 -(l_{max}+1)Nd_R}}  \\
        \varepsilon_{\rm hyb} &=  \frac{ky + k\log_2(Nk)}{2^{b'_H-1}} \,\,\Xi\left(\left\lfloor\frac{N}{2}\right\rfloor +1, N(1-d_R) ,\exp\left(- \frac{k}{8(l_{max}+1)^2}\right) \right),
    \end{align}
    and
    \begin{align}
        \Xi(k,n,p):= \sum_{j=k}^{n} \binom{n}{j} p^j (1-p)^{n-j}. \label{Xi-Andersson-Appendix}
    \end{align}
\end{lemma*}

\begin{proof}
According to the description of Protocol~\ref{QDSprot-Andersson}, if a signature is verified at level $l$, then it is also verified also at any lower level:
\begin{align}
    &\mathrm{Ver}_{j,l}(Doc,Sig)=\mathrm{True} \implies \mathrm{Ver}_{j,l'}(Doc,Sig)=\mathrm{True}, \,\, \forall \, l'<l.
\end{align}
By the contrapositive we have that:
\begin{align}
    \mathrm{Ver}_{j,l'}(Doc,Sig)=\mathrm{False} \implies \mathrm{Ver}_{j,l}(Doc,Sig)=\mathrm{False} \quad\forall\, l>l'.
\end{align}
Thus, we can restrict without loss of generality the condition in Definition~\ref{def:transferability-Andersson} to verifying that
\begin{align}
        \Pr&\left[\exists\, P_i,P_j \notin C : \,\mathrm{Ver}_{i,l}(Doc,Sig)=\mathrm{True} \, \wedge   \,\mathrm{Ver}_{j,l-1}(Doc,Sig)=\mathrm{False}\right] \leq \varepsilon_{\rm transf}, \label{transf-toprove}
\end{align}
where we consider the largest possible coalition $C$ of dishonest users, formed by the dishonest sender and $Nd_R$ dishonest receivers.

To start with, we fix the pair of honest receivers that are targeted by the coalition to be $P_i$ and $P_j$, respectively. Nevertheless, the attack is successful whenever any pair of honest receivers disagree on the verification outcome; we cover this case at the end of the proof. Moreover, we start by considering the possible attacks from the coalition $C$ that do not involve attacking the authenticated channels and are instead based on distributing invalid hash functions to honest receivers.

The probability of a successful attack of the coalition on the pair $P_i,P_j$ is given by: 
\begin{align}
    \Pr[\mbox{attack on }P_i,P_j] &= \Pr\left[\mathrm{Ver}_{i,l}(Doc,Sig)=\mathrm{True} \, \wedge   \,\mathrm{Ver}_{j,l-1}(Doc,Sig)=\mathrm{False}\right] \nonumber\\
    &=\Pr\left[\sum_{r=1}^N T^{Doc}_{i,r,l}>\frac{N}{2} +(l+1)Nd_R \, \wedge \, \sum_{r=1}^N T^{Doc}_{j,r,l-1} \leq \frac{N}{2} + l N d_R\right], \label{transf-condition-proof}
\end{align}
where we used the definitions of the verification function and of $\delta_l$. Now, we observe that for the tests $T^{Doc}_{j,r,l-1}$ where $r\in C$, the coalition can force the test to not pass. Indeed, a dishonest receiver $r$ can forward invalid hash functions to $P_j$ such that the number of discrepancies observed by $P_j$ exceeds the threshold $s_{l-1} k$, forcing $T^{Doc}_{j,r,l-1}=0$. At the same time, the coalition can behave honestly with respect to receiver $P_i$, making sure that their tests are passed: $T^{Doc}_{i,r,l}=1$. Since there are $Nd_R$ dishonest receivers, there can be at most $Nd_R$ tests that are passed by $P_i$ and not passed by $P_j$. Since this occurs deterministically and cannot be avoided, the condition to be checked in \eqref{transf-condition-proof} updates to:
\begin{align}
    \Pr[\mbox{attack on }P_i,P_j] &= \Pr\left[\sum_{h=1}^{N(1-d_R)} T^{Doc}_{i,h,l}>\frac{N}{2} + lNd_R \, \wedge \, \sum_{h=1}^{N(1-d_R)} T^{Doc}_{j,h,l-1} \leq \frac{N}{2} + l N d_R\right], \label{transf-condition-proof2}
\end{align}
where now the sums only run over the set of honest receivers, $h \notin C$. At this point, the coalition's ability to steer the result of a test where the hash functions originate from an honest receiver, $h$, is reduced, but not null. In particular, to each honest receiver $h$, the dishonest sender can provide a set of $Nk$ hash functions, some of which are correct while others are invalid (in the sense that their hashed values are different from the tags). Then, depending on the random samples of $k$ hash functions drawn by $h$ and sent to participants $P_i$ and $P_j$, the respective tests $T^{Doc}_{i,h,l}$ and $T^{Doc}_{j,h,l-1}$ might fail or pass.

In order to find the optimal strategy for the coalition, we consider a necessary condition for the event in \eqref{transf-condition-proof2} to occur. Namely, that there exists at least one honest receiver, $\bar{h}$, such that $P_i$'s test is passed while $P_j$'s test fails. Note that this condition may also be sufficient since the event in \eqref{transf-condition-proof2} occurs as soon as the number of passed tests between $P_i$ and $P_j$ differs by one. Since the probability of a necessary condition is larger than the probability of the original event, we have the upper bound:
\begin{align}
    \Pr[\mbox{attack on }P_i,P_j] &\leq \Pr\left[\exists\,\bar{h}\notin C:\,T^{Doc}_{i,\bar{h},l} =1 \, \wedge\, T^{Doc}_{j,\bar{h},l-1} =0 \right] \nonumber\\
    &\leq N(1-d_R) \Pr\left[T^{Doc}_{i,h,l} =1 \, \wedge\, T^{Doc}_{j,h,l-1} =0 \right] \nonumber\\
    &\leq N(1-d_R) \min\left\lbrace\Pr\left[T^{Doc}_{i,h,l} =1  \right] , \Pr\left[T^{Doc}_{j,h,l-1} =0 \right] \right\rbrace, \label{transf-condition-proof3}
\end{align}
where in the second inequality we used the union bound and assumed without loss of generality that the probability of mismatching outcomes for participants $P_i$ and $P_j$ is independent of the choice of honest receiver $h$.

As mentioned earlier, the coalition can provide participant $h$ with a set of hash functions where some of the functions are invalid, meaning that their hashed values would differ from the tags sent by the sender. Regardless of the strategy, since participant $h$ randomly samples the sets $F_{h \to i}$ and $F_{h \to j}$ to be sent to $P_i$ and $P_j$, respectively, the expected number of mismatches between the tags and the hashed values observed by $P_i$ and by $P_j$ coincides. We define the expected fraction of mismatches observed by $P_i$ and $P_j$ to be $p_e$:
\begin{align}
    p_e &= \mbox{E}\left[\sum_{r\in R_{h \to i}} \frac{g(t_r,f_r(Doc))}{k}\right] \\
    &= \mbox{E}\left[\sum_{r\in R_{h \to j}} \frac{g(t_r,f_r(Doc))}{k}\right],
\end{align}
and we assume it to be fixed by the optimal strategy found by the coalition (in other words, $p_e$ is not a random variable). Now, we consider different choices for the value of $p_e$ and compute the bound in \eqref{transf-condition-proof3} for each choice.

\begin{itemize}
    \item $p_e < s_l < s_{l-1}$ \quad In this case, the expected fraction of errors is below the thresholds used in both tests. Therefore, it is likely that both tests are passed. This implies that we can trivially bound the probability of passing the test in $P_i$:
    \begin{align}
        \Pr\left[T^{Doc}_{i,h,l} =1  \right] \leq 1. \label{b1}
    \end{align}
    On the contrary, we can use Hoeffding's inequality to bound the probability that the test in $P_j$ fails:
    \begin{align}
        \Pr\left[T^{Doc}_{j,h,l-1} =0 \right] &= \Pr\left[\sum_{r\in R_{h \to j}} g(t_r,f_r(Doc))  \geq s_{l-1} k \right] \nonumber\\
        &= \Pr\left[\sum_{r\in R_{h \to j}} g(t_r,f_r(Doc)) -p_e k \geq (s_{l-1} -p_e) k \right] \nonumber\\
        &\leq \exp\left(-2 k (s_{l-1} -p_e)^2\right). \label{b2}
    \end{align}
    By combining the bounds in \eqref{b1} and \eqref{b2} and the condition on $p_e$, we obtain the following upper bound on the success probability of the attack from \eqref{transf-condition-proof3}:
    \begin{align}
    \Pr[\mbox{attack on }P_i,P_j] &\leq N(1-d_R) \exp\left(-2 k (s_{l-1} -s_l)^2\right). \label{case1}
    \end{align}
    
    \item $p_e > s_{l-1} > s_l$ \quad In this case, the expected fraction of errors exceeds both thresholds, thus likely causing both tests to fail. We thus obtain the following bounds:
    \begin{align}
        \Pr\left[T^{Doc}_{j,h,l-1} =0 \right] &\leq 1, \\
        \Pr\left[T^{Doc}_{i,h,l} =1  \right] &= \Pr\left[\sum_{r\in R_{h \to i}} g(t_r,f_r(Doc)) < s_{l} k \right] \nonumber\\
        &= \Pr\left[p_e k - \sum_{r\in R_{h \to i}} g(t_r,f_r(Doc)) > (p_e - s_{l} )k \right] \nonumber\\
        &\leq  \exp\left(-2 k (p_e -s_l)^2\right),
    \end{align}
    where the second bound is again obtained by applying Hoeffding's inequality. By combining the above bounds and the condition on $p_e$, we obtain the following upper bound on the success probability of the attack:
    \begin{align}
    \Pr[\mbox{attack on }P_i,P_j] &\leq N(1-d_R) \exp\left(-2 k (s_{l-1} -s_l)^2\right). \label{case2}
    \end{align}

    \item $s_l < p_e < s_{l-1}$ \quad In this case, it is likely that the test in $P_i$ fails while the test in $P_j$ succeeds. Thus, we can use Hoeffding's inequality to bound both probabilities in \eqref{transf-condition-proof3}:
    \begin{align}
        \Pr\left[T^{Doc}_{i,h,l} =1  \right] &\leq \exp\left(-2 k (p_e -s_l)^2\right) \\
        \Pr\left[T^{Doc}_{j,h,l-1} =0 \right] &\leq \exp\left(-2 k (s_{l-1} -p_e)^2\right).
    \end{align}
    By combining the above bounds in \eqref{transf-condition-proof3}, we obtain the following upper bound on the attack success probability:
    \begin{align}
    \Pr[\mbox{attack on }P_i,P_j] &\leq N(1-d_R) \min\left\lbrace \exp\left(-2 k (p_e -s_l)^2\right), \exp\left(-2 k (s_{l-1} -p_e)^2\right) \right\rbrace .
    \end{align}
    Since the optimal strategy by the coalition aims at maximizing the attack success probability, we choose $p_e$ to be equidistant from $s_l$ and $s_{l-1}$, i.e. $p_e=(s_l + s_{l-1})/2$. Indeed, this choice maximizes the right hand side of the last expression. We obtain:
    \begin{align}
    \Pr[\mbox{attack on }P_i,P_j] &\leq N(1-d_R) \exp\left(- \frac{k}{2} (s_{l-1} -s_l)^2\right) . \label{case3}
    \end{align} 
\end{itemize}

By comparing the bounds \eqref{case1}, \eqref{case2}, and \eqref{case3} obtained for various possible choices of the value of $p_e$, we observe that the largest attack probability on a given pair of participants $P_i,P_j$ is obtained in \eqref{case3} for $p_e=(s_l + s_{l-1})/2$. Thus, we obtain the following upper bound on the probability of a successful attack on the fixed pair $P_i,P_j$ of honest participants, given that the coalition $C$ performs the attack with invalid hash functions provided to honest receivers:
\begin{align}
    \Pr[\mbox{attack on }P_i,P_j] &\leq \varepsilon_{i,j}, \label{attackPiPj-noauth}
\end{align}
where we employed the optimal choice for the spacing of the $s_l$ variables, $s_{l-1}-s_l=\Delta s\approx [2(l_{max}+1)]^{-1}$,
\begin{align}
    \varepsilon_{i,j} = N(1-d_R) \exp\left(- \frac{k}{8(l_{max}+1)^2}\right). \label{attackPiPj-noauth-prob}
\end{align}

At this point, we turn to consider the attack on the transferability property between $P_i$ and $P_j$ which is enabled by imperfect authenticated channels, while assuming that all the hash functions distributed by the sender are correct. Specifically, the coalition can either attack the authenticated channels used to forward the $\{Doc,Sig\}$ pair in the messaging stage, or the authenticated secret channels used to shuffle the hash functions in the distribution stage. For the former, a successful attack would indeed cause a rejection but it comes at the cost of having one of the participants verifying a different pair $\{Doc',Sig'\}$. Since the transferability definition requires both parties to verify the same document-signature pair, this scenario is not relevant for the proof (although  causing a transferability issue in practice). In the latter attack, the coalition attempts to alter the hash functions sent to party $P_j$ over the authenticated secret channels by honest receivers, such that $P_j$ rejects the signature. Meanwhile, the hash functions destined to $P_i$ are left untouched, hence the event $\mathrm{Ver}_{i,l}(Doc,Sig)=\mathrm{True}$ occurs with certainty. The probability that the described attack is successful is given by:
\begin{align}
    \Pr\left[\mbox{attack on }P_j\right] &= \Pr\left[\mathrm{Ver}_{j,l-1}(Doc,Sig)=\mathrm{False}\right] \nonumber\\
    &=\Pr\left[\sum_{h=1}^{N(1-d_R)-1} T^{Doc}_{j,h,l-1} \leq \frac{N}{2} + l N d_R-1\right],  \label{attackPj-auth}
\end{align}
where we already accounted for the fact that the $Nd_R$ dishonest receivers will send invalid hash functions to $P_j$ and that the hash functions kept by $P_j$, i.e. those in the set $F_{j \to j}$, are necessarily correct, thus causing $T_{j,j,l-1}^{Doc}=1$. 

From \eqref{attackPj-auth}, we deduce that $P_j$ rejects the signature if the number of tests that are not passed, $n_t$, satisfies:  $N(1-d_R)-1 - n_t \leq N/2 + l_{max} N d_R-1$, where we considered the minimum requirement on $n_t$ by taking $l=l_{max}$. By solving for $n_t$, we obtain the following number of non-passed tests for rejecting the signature:
\begin{align}
    n_t = \ceil{N/2 -(l_{max}+1)Nd_R}. \label{nt}
\end{align}

This implies that the coalition must successfully corrupt the hash functions sets $F_{h \to j}$ sent by $n_t$ honest receivers over authenticated secret channels. Now, recall from Sec.~\ref{sec:notation} that the authenticated channels are implemented via WC authentication with key recycling and with the hash family $\mathcal{F}_{\rm AXU}$. Then, the probability of corrupting one authenticated message of length $b_M$ is $b_M/2^{b'_H-1}$ according to Table~\ref{tab:efficient-hashing}. Note, however, that a single failed authentication (caused by a corruption attempt) causes the protocol to abort. Then, the probability that the coalition successfully corrupts enough hash function sets to cause $P_j$'s rejection without causing an abortion reads:
\begin{align}
    \Pr\left[\mbox{attack on }P_j\right] &= \left(\frac{ky + k\log_2(Nk)}{2^{b'_H-1}}\right)^{n_t} =: \varepsilon_{\rm auth},  \label{attackPj-auth-prob}
\end{align}
where $ky + k\log_2(Nk)$ is the length of the messages exchanged by the receivers in the distribution stage.

Now, for each given pair $P_i,P_j$, the coalition might decide to perform the first attack we analyzed, whose probability of success is given by \eqref{attackPiPj-noauth-prob}, or the second attack with probability of success \eqref{attackPj-auth-prob}, or a combination of the two. Combining the two attacks means that the verification tests at $P_j$ may fail either due to incorrect hash functions received from honest receivers, or due to hash functions sent by honest receivers which are corrupted on their way to $P_j$. The necessary events for the hybrid attack to take place are that $P_i$ successfully verifies the signature and that at least one successful attack is performed on the authenticated channels. Thus, the success probability of the hybrid approach is bounded by:
\begin{align}
    \Pr\left[\mbox{hybrid attack}\right] &\leq  \Pr\left[\mathrm{Ver}_{i,l}(Doc,Sig)=\mathrm{True}\right]  \frac{ky + k\log_2(Nk)}{2^{b'_H-1}} \nonumber\\
    &=  \Pr\left[\sum_{h=1}^{N(1-d_R)} T^{Doc}_{i,h,l}>\frac{N}{2} + lNd_R \right]  \frac{ky + k\log_2(Nk)}{2^{b'_H-1}} \nonumber\\
    &\leq  \frac{ky + k\log_2(Nk)}{2^{b'_H-1}} \Xi\left(\left   \lfloor\frac{N}{2} + lNd_R \right\rfloor + 1,N(1-d_R),e^{-\frac{k}{2} (s_{l-1} -s_l)^2}\right) \nonumber\\
    &\leq \varepsilon_{\rm hyb} ,
\end{align}
with
\begin{align}
    \varepsilon_{\rm hyb} := \frac{ky + k\log_2(Nk)}{2^{b'_H-1}} \Xi\left(\left\lfloor\frac{N}{2}\right\rfloor +1 , N(1-d_R) ,\exp\left(- \frac{k}{8(l_{max}+1)^2}\right) \right). \label{attack-hybrid}
\end{align}
In the first equality we accounted for the fact that the dishonest receivers will provide $P_i$ with correct hash functions. In the second inequality, we considered the probability of passing a single test at $P_i$ as per derivation of \eqref{case3} and used the function \eqref{Xi-Andersson-Appendix} for the probability of at least $\frac{N}{2} + lNd_R$ successes out of $N(1-d_R)$ attempts. Finally, in the last inequality we chose the smallest value\footnote{We remark that, according to the transferability definition, the smallest allowed value for $l$ is $l=1$. Here, however, we choose $l=0$ so that the resulting bound is also reusable in the proof against repudiation.} for $l$ and adopted the optimal spacing $s_{l-1}-s_l=\Delta s=[2(l_{max}+1)]^{-1}$.

Since the coalition aims at maximizing its attack success probability, for each fixed pair $P_i,P_j$ the coalition will perform the type of attack with the highest probability of success, which is given by:
\begin{align}
    \max\{\varepsilon_{i,j},\varepsilon_{\rm auth},\varepsilon_{\rm hyb}\}.
\end{align}
Recall that the total success probability in \eqref{transf-toprove} refers to any pair of honest participants. There are in total $\binom{N(1-d_R)}{2}$ different pairs of honest participants. By employing the union bound, the relevant probability can be bounded by:
\begin{align}
    \Pr&\left[\exists\, P_i,P_j \notin C : \,\mathrm{Ver}_{i,l}(Doc,Sig)=\mathrm{True} \, \wedge   \,\mathrm{Ver}_{j,l-1}(Doc,Sig)=\mathrm{False}\right] \leq \binom{N(1-d_R)}{2} \max\{\varepsilon_{i,j},\varepsilon_{\rm auth},\varepsilon_{\rm hyb}\} \label{transf-proof-result}
\end{align}
where the epsilon variables are given in \eqref{attackPiPj-noauth-prob}, \eqref{attackPj-auth-prob} and \eqref{attack-hybrid}, respectively. This concludes the proof.
\end{proof}

\begin{lemma*} 
    Protocol~\ref{QDSprot-Andersson} is $\varepsilon_{\rm rep}$-secure against repudiation according to Definition~\ref{def:repudiation-Andersson}, with:
    \begin{align}
        \varepsilon_{\rm rep}= \binom{N(1-d_R)}{2}\max\{\varepsilon_{i,j},\varepsilon_{\rm auth},\varepsilon_{\rm hyb}\},
    \end{align}
    where $\varepsilon_{i,j},\varepsilon_{\rm auth}$ and $\varepsilon_{\rm hyb}$ are given in \eqref{attackPiPj-noauth-prob}, \eqref{attackPj-auth-prob} and \eqref{attack-hybrid}, respectively.
\end{lemma*}

\begin{proof}
The proof can be reduced to a special case of the proof of transferability (see proof of Lemma~\ref{lm:transferability&rep-Andersson}). According to the repudiation definition (Definition~\ref{def:repudiation-Andersson}), we need to bound the following probability:
\begin{align}
    p_{\rm rep} = \Pr&\left[\exists\, P_i\notin C : \,\mathrm{Ver}_{i,l}(Doc,Sig)=\mathrm{True} \,\wedge \,\mathrm{MV}(Doc,Sig) = \mathrm{Invalid}\right].
\end{align}

Here, the attack is successful if an honest receiver validates the signature while the majority of receivers (specifically, $\geq \ceil{N/2} $) deems it invalid at verification level $l=-1$. With a coalition of $N d_R = \omega -1$ dishonest receivers, we can already be certain that $\omega -1$ participants will deem the signature as invalid in the majority vote. This means that the attack is successful when at least $N_{HR} :=\ceil{N/2} +1-\omega$ honest receivers reject the signature at verification level $l=-1$. By taking into account that the number of dishonest participants is strictly bounded by: $\omega <(N+1)/2$, the number of honest receivers that need to reject the signature for the attack to be successful is at least $N_{HR}  > \ceil{N/2} - N/2 + 1/2$, which is equivalent to requesting:
\begin{align}
    N_{HR} \geq \left\lbrace \begin{array}{ll}
        2 & \mbox{for $N$ odd} \\
        1 & \mbox{for $N$ even}. 
    \end{array} \right.
\end{align} 
In other words, a precondition for the attack being successful for any $N$ is that at least one honest receiver rejects the signature at verification level $l=-1$. Then, the probability to be bounded reads:
\begin{align}
         p_{\rm rep} &\leq \left[\exists\, P_i,P_j\notin C : \,\mathrm{Ver}_{i,l}(Doc,Sig)=\mathrm{True} \,\wedge \,\mathrm{Ver}_{j,-1}(Doc,Sig)=\mathrm{False}\right].
\end{align}
Now, we use the fact pointed out in \eqref{verfication-cascade} to upper bound the last expression as follows:
\begin{align}
        p_{\rm rep} &\leq \left[\exists\, P_i,P_j\notin C : \,\mathrm{Ver}_{i,0}(Doc,Sig)=\mathrm{True} \,\wedge \,\mathrm{Ver}_{j,-1}(Doc,Sig)=\mathrm{False}\right] \nonumber\\
        &\leq \binom{N(1-d_R)}{2} \max\{\varepsilon_{i,j},\varepsilon_{\rm auth},\varepsilon_{\rm hyb}\},
\end{align}
where we used \eqref{transf-proof-result} in the last inequality, since it also valid for the special case $l=0$. This concludes the proof.
\end{proof}

\begin{lemma*} 
    Protocol~\ref{QDSprot-Andersson} is $\varepsilon_{\rm for}$-secure against forgery according to Definition~\ref{def:forgery-Andersson}, with:
    \begin{align}
        \varepsilon_{\rm for}= (N-\omega)\,\,\Xi(\floor{N/2},N-\omega,p_t),
    \end{align}
    where $p_t$ is defined as:
    \begin{align}
        p_t =\Xi(\floor{k(1-s_0)}+1,k,2^{1-b_H}),
    \end{align}
    and
    \begin{align}
        \Xi(k,n,p):= \sum_{j=k}^{n} \binom{n}{j} p^j (1-p)^{n-j}. \label{Xi-Andersson-Appendix2}
    \end{align}
    Moreover, the forgery attack based on the security loophole discussed in Ref.~\cite{QDS-Kiktenko}, where a dishonest receiver forges the pair $\{Doc',Sig'\}$ and invokes the dispute resolution method for the other parties to accept the pair, succeeds with probability:
    \begin{align}
        p_{\rm attack} = \Xi(\floor{N/2}+1 - \omega ,N-\omega,p_{-1}),
    \end{align}
    with:
    \begin{align}
        p_{-1} = \Xi(\floor{N/2}+1 - \omega ,N-\omega,p_t).
    \end{align} 
\end{lemma*}

\begin{proof}
A successful forgery attack occurs when a coalition $C$ of $\omega$ dishonest receivers, which does not include the sender, is provided with a valid $\{Doc,Sig\}$ pair and finds a new pair $\{Doc',Sig'\}$ (with $Doc' \neq Doc$) such that at least one of the $N-\omega$ honest receivers accepts it at the lowest verification level, $l=0$.

We first consider the case where the coalition tries to deceive a fixed receiver, $P_i$. Then, we are interested in bounding the following probability:
\begin{align}
    \Pr\left[\mathrm{Ver}_{i,0}(Doc',Sig')=\mathrm{True}  \right] &= \Pr\left[\sum_{j=1}^N T^{Doc'}_{i,j,0} \geq \floor{N \delta_0} + 1  \right].
\end{align}
By using the fact that $\delta_0 =1/2 + d_R$ and that $d_R=(\omega-1)/N$, we can write:
\begin{align}
    \Pr\left[\mathrm{Ver}_{i,0}(Doc',Sig')=\mathrm{True}  \right] &= \Pr\left[\sum_{j=1}^N T^{Doc'}_{i,j,0} \geq \left\lfloor\frac{N}{2}\right\rfloor  +\omega  \right]. \label{forgery-toprove}
\end{align}
Now, for all the $P_j \in C$, the test $T^{Doc'}_{i,j,0}$ can be forced to pass. Indeed, the coalition knows the hash functions in $F_{j \to i}$ and thus can choose the tags contained in $Sig'$, corresponding to the positions $R_{j\to i}$, to be the hashed values of $Doc'$ when using the functions in $F_{j \to i}$. This will guarantee that $T^{Doc'}_{i,j,0}=1$ occurs with certainty. Then, we can update the probability in \eqref{forgery-toprove} as follows:
\begin{align}
    \Pr\left[\mathrm{Ver}_{i,0}(Doc',Sig')=\mathrm{True}  \right] &= \Pr\left[\sum_{h=1}^{N-\omega} T^{Doc'}_{i,h,0} \geq \left\lfloor\frac{N}{2}\right\rfloor \right], \label{forgery-toprove2}
\end{align}
meaning that the attack is successful if at least $\left\lfloor N/2 \right\rfloor$ additional tests originating from honest receivers $h \notin C$ are passed. Let $p_t$ be the probability that the test based on the hash functions sent by receiver $h$ is passed:
\begin{align}
    p_t:= \Pr\left[T^{Doc'}_{i,h,0} = 1\right] \label{pt}.
\end{align}
Then, assuming w.l.o.g. that the probability of passing each test is independent for different receivers $h$, we can express the probability in \eqref{forgery-toprove2} as follows:
\begin{align}
    \Pr\left[\mathrm{Ver}_{i,0}(Doc',Sig') = \mathrm{True}  \right] &=  \Xi(\floor{N/2},N-\omega,p_t), \label{forgery-toprove3}
\end{align}
where $\Xi(k,n,p)$, defined in \eqref{Xi-Andersson-Appendix2}, is the function returning the probability of at least $k$ successes out of $n$ trials with success probability $p$ for each trial.

We now consider the general case of the coalition trying to deceive at least one of the $N-\omega$ honest receivers. The relevant probability is thus:
\begin{align}
        \Pr\left[\exists\, P_i \notin C : \,\mathrm{Ver}_{i,0}(Doc',Sig')=\mathrm{True}  \right] \leq (N-\omega) \Pr\left[\mathrm{Ver}_{i,0}(Doc',Sig')=\mathrm{True}  \right] \label{forgery-unionbound},
    \end{align}
where we employed the union bound in the inequality. Then, by combining \eqref{forgery-unionbound} with \eqref{forgery-toprove3} we obtain:
\begin{align}
    \Pr\left[\exists\, P_i \notin C : \,\mathrm{Ver}_{i,0}(Doc',Sig')=\mathrm{True}  \right] \leq (N-\omega) \,\,\Xi(\floor{N/2},N-\omega,p_t),
\end{align}
which is the claim of the Lemma.

We now derive an explicit expression for $p_t$. Participant $P_i$ deems the test $T^{Doc'}_{i,h,0}$ as passed (at verification level $l=0$) if the number of mismatches between the tags in $Sig'$ and the hashes obtained by applying the functions in $F_{h \to i}$ to $Doc'$ is smaller than $k s_0$:
\begin{align}
    p_t = \Pr\left[T^{Doc'}_{i,h,0} = 1\right] = \Pr\left[\sum_{r \in R_{h \to i}} g(t_r',f_r(Doc')) < k s_0 \right].
\end{align}
In order for the coalition to append correct tags in $Sig'$, they need to know the hash functions in $F_{h \to i}$. However, this is not possible since $h$ is an honest receiver and the sender is not part of the coalition. Let $f_1, f_2, \dots, f_k$ be the $k$ hash functions in $F_{h \to i}$. For each hash function, the coalition knows a valid message-tag pair from the knowledge of $\{Doc,Sig\}$ and wishes to find another valid pair. In particular, the coalition knows that $Doc$ and $t_1$ are such that $f_1(Doc)=t_1$ and wants to find $t_1'$ such that $f_1(Doc')=t_1'$. Since the $\varepsilon$-ASU$_2$ family $\mathcal{F}_{\rm ASU}$ adopted by the protocol has security parameter $\varepsilon=2^{1-b_H}$ (see Appendix~\ref{app:efficient-hashing}), the probability that the coalition correctly guesses the tag $t_1'$ is given by $2^{1-b_H}$. Since correctly guessing the tag for each hash function in $F_{h \to i}$ are independent events, the probability of correctly guessing more than $k- k s_0$ tags out of $k$ attempts is:
\begin{align}
    p_t = \Pr\left[T^{Doc'}_{i,h,0} = 1\right] = \Xi(\floor{k(1-s_0)} +1 ,k,2^{1-b_H}),
\end{align}
which is the expression provided in the claim of the Lemma.

Finally, we compute the success probability ($p_{\rm attack}$) of a forgery attempt where the forger invokes the dispute resolution method by falsely claiming that they verified the pair $\{Doc',Sig'\}$ at level $l=0$ (this attack was first pointed out in Ref.~\cite{QDS-Kiktenko}). In this case, the forger aims at making other honest receivers accept the signature at level $l=-1$, i.e., at a lower level compared to the standard definition of security against forgery, such that the outcome of the dispute resolution is to accept the pair: $\mathrm{MV}(Doc',Sig') = \mathrm{Valid}$.

Let $P_i \notin C$ be an honest receiver. Then, the probability that $P_i$ accepts the forged signature is:
\begin{align}
    p_{-1} &:= \Pr\left[\mathrm{Ver}_{i,-1}(Doc',Sig')=\mathrm{True}  \right] \nonumber\\
    &= \Pr\left[\sum_{j=1}^N T^{Doc'}_{i,j,-1} \geq \floor{N \delta_{-1}} + 1  \right] \nonumber\\
    &= \Pr\left[\sum_{h=1}^{N-\omega} T^{Doc'}_{i,h,-1} \geq \left\lfloor \frac{N}{2} \right\rfloor + 1 - \omega  \right],
\end{align}
where we assumed that the dishonest receivers will automatically force the test to be passed ($ T^{Doc'}_{i,j,-1} =1$). By again using the fact that $p_t$ is the probability that a single test $ T^{Doc'}_{i,h,-1}$ is passed, we have the following probability that an honest receiver accepts the forged signature at level $l=-1$:
\begin{align}
    p_{-1} &=\Xi(\floor{N/2}+1 -\omega,N-\omega,p_t) .
\end{align}
Now consider that, to enforce $\mathrm{MV}(Doc',Sig') = \mathrm{Valid}$, the malicious coalition needs a total of $\floor{N/2} + 1$ receivers accepting the pair. Since $\omega$ of them will surely accept, only $\floor{N/2} + 1 - \omega$ honest receivers out of $N- \omega$ need to accept the pair. This occurs with probability:
\begin{align}
    p_{\rm attack} = \Xi(\floor{N/2}+1 - \omega ,N-\omega,p_{-1}),
\end{align}
as reported in the claim of the Lemma. This concludes the proof.
\end{proof}

\section{Bracha's reliable broadcast} \label{app:BRB}

In this appendix we provide details on the reliable broadcast protocol that we adopt in our ERS protocol in Sec.~\ref{sec:ERS}.

We consider a network of $N$ parties given by $\mathcal{P}=\{P_1,\dots,P_N\}$, with up to $t<N/3$ faulty or dishonest parties. For communicating in this network, we assume a practical asynchronous message-passing model with authenticated channels, in which messages sent between honest parties are eventually delivered (but in an arbitrary order), while dishonest parties may arbitrarily delay or omit their own messages. Message delays are assumed to be unbounded but finite.

In a reliable broadcast \cite{ReliableBroadcast-Bracha}, a designated party $P_s \in \mathcal{P}$, called the sender, broadcasts a request $m$ to all other parties. The goal of the protocol is that all parties, eventually, deliver $m$. In other words, the protocol should ensure that the message $m$ originating from $P_s$ is shared among all participants.

More formally, a reliable broadcast should satisfy the following conditions \cite{BRBdef}.
\begin{definition} \label{def:BRB}
    A reliable broadcast protocol satisfies:
    \begin{itemize}
        \item \textbf{Validity}\quad If an honest sender broadcasts $m$, then all honest parties eventually deliver $m$.
        \item \textbf{Consistency}\quad If an honest party delivers $m$ and another honest party delivers $m'$, then $m=m'$.
        \item \textbf{Integrity}\quad Every honest party delivers at most one request.
        \item \textbf{Totality}\quad If an honest party delivers a request, then eventually all other honest parties deliver a request.
    \end{itemize}
\end{definition}

We now present Bracha's reliable broadcast (BRB) protocol \cite{ReliableBroadcast-Bracha} that is adopted by our ERS protocol in Sec.~\ref{sec:ERS}. All messages in the protocol are assumed to be authenticated. Moreover, when a party sends a message directed to all parties, it also sends that message to themselves.\\

\begin{protocol} \label{BRBprot} \caption{Bracha's reliable broadcast (BRB)}
\begin{enumerate}[wide, labelwidth=!, labelindent=0pt]

\item Only $P_s$: The sender sends (SEND, $m$) to all parties in $\mathcal{P}$.
\item All parties in $\mathcal{P}$:

\begin{enumerate}[label*=\arabic*.]
    \item upon receiving a message (SEND, $m$) from $P_s$: send (ECHO, $m$) to all in $\mathcal{P}$.
    \item upon receiving $\ceil{\frac{N+t+1}{2}}$ messages (ECHO, $m$) and not having sent a READY message: send message (READY, $m$) to all in $\mathcal{P}$.
    \item upon receiving $t+1$ messages (READY, $m$) and not having sent a READY message: send message (READY, $m$) to all in $\mathcal{P}$.
    \item upon receiving $2t+1$ messages (READY, $m$): deliver $m$.
\end{enumerate}

\end{enumerate}
\end{protocol}
\vspace{2ex}

The BRB protocol is proved to satisfy Definition~\ref{def:BRB} for $N>3t$. As such, it guarantees that either all honest parties eventually deliver the same message $m$, or no honest party delivers any value. The latter stall condition can be caused by a faulty or dishonest sender who sends different messages $m$ and $m'$ in the SEND phase, thereby preventing the READY conditions from ever being satisfied. 

To handle the stall condition in practice, each party may employ a local timeout mechanism: if no delivery occurs and no further protocol progress is observed after a reasonable time period, the party outputs ABORT. Since messages between honest parties are eventually delivered, any execution in which BRB would deliver a value will do so before each local timeout expires, whereas executions in which the protocol would stall indefinitely result in all honest parties eventually aborting.

In summary, the practical version of the BRB protocol with the additional timeout, when integrated in our ERS protocol (Protocol~\ref{ERSprot}), ensures that either all honest parties are provided with the same $\{Doc,Sig\}$ pair, or all honest parties abort the ERS protocol.

In order to quantify the resources required by the BRB protocol when benchmarking our ERS protocol against previous protocols, we report the total number of authenticated messages sent in a BRB execution (modulo the messages sent to one self): $(2N+1)(N-1)$.

\section{Security proof of Protocol~\ref{ERSprot}} \label{app:security-ERS}

In this Appendix we prove the IT security of the ERS protocol introduced in Sec.~\ref{sec:ERS}.

\begin{lemma*} 
    Protocol~\ref{ERSprot} is $\varepsilon_{\rm for}$-secure against forgery according to Definition~\ref{def:forgery-ERS}, with:
    \begin{align}
        \varepsilon_{\rm for}= 2^{1-b_H} .
    \end{align}
\end{lemma*}

\begin{proof}
    A successful forgery attack occurs when a coalition $C$ of dishonest receivers, which does not include the sender, is provided with a valid $\{Doc,Sig\}$ pair and finds a new pair $\{Doc',Sig'\}$ (with $Doc' \neq Doc$) such that at least one of the honest receivers accepts the signature.

    We consider the worst case scenario where $P_1$ is part of the coalition $C$. In this way, $P_1$ receives the valid $\{Doc,Sig\}$ pair directly from the sender, where $Sig = (t||h_s)$ and $h_s = h \oplus X_s$, and can use it to forge the pair $\{Doc',Sig'\}$. Note that there is no way for the coalition $C$ to learn $X_s = \oplus_{i=1}^N X_i$ as far as there is at least one honest receiver. This means that $C$ cannot recover the hash function $f_h$ used by sender from the value of $h_s$. More precisely, the coalition $C$ has no information on the random choice of $f_h \in \mathcal{F}_{\rm ASU}$ made by the sender.

    We consider two attack strategies by the coalition. In the first strategy, the coalition ignores the valid $\{Doc,Sig\}$ pair from the sender and forges a new pair $\{Doc',Sig'\}$ by randomly picking a function $f_{h'} \in \mathcal{F}_{\rm ASU}$ and defining $Sig'=(f_{h'}(Doc')||h'\oplus r)$, where $r$ is a string of $y$ bits and represents a random guess of the value of $X_s$.
    
    The coalition, through $P_1$, broadcasts the forged pair to all honest receivers via the BRB protocol and follows the remaining steps of the protocol truthfully. Then, each honest receiver will validate the signature if and only if $r=X_s$. This occurs with probability $2^{-y}$, which represents the success probability of the first attack strategy.

    Note that the coalition cannot deliver different forged pairs to distinct honest receivers, since the BRB protocol ensures that either all honest receivers obtain the same pair or everyone aborts.

    In the second attack strategy, the coalition uses the valid $\{Doc,Sig\}$ pair to forge the new pair $\{Doc',Sig'\}$ and share it with the other honest receivers via the BRB protocol. In particular, the coalition knows a message-tag pair $(Doc,t)$ (related by an unknown hash function $f_h \in \mathcal{F}_{\rm ASU}$) and tries to find another pair $(Doc',t')$, with $Doc' \neq Doc$, such that $f_h(Doc')=t'$. Then, the new signature is set by the coalition to: $Sig'=(t'|| h_s)$. This attack is equivalent to attacking a WC authentication scheme based on the $\varepsilon$-ASU$_2$ family $\mathcal{F}_{\rm ASU}$ defined in Appendix~\ref{app:efficient-hashing}. Then, the success probability of the second attack strategy is $\varepsilon=2^{1-b_H}$ \cite{Kiktenko_2020}.

    By combining the two attack strategies, the forgery attack succeeds with probability at most:
    \begin{align}
        \varepsilon_{\rm for}=\max \left\lbrace 2^{-y}, 2^{1-b_H}\right\rbrace
    \end{align}
    By recalling the explicit expression of $y$ from Appendix~\ref{app:efficient-hashing}, we observe that $y>b_H -1$ and thus $\varepsilon_{\rm for}=2^{1-b_H}$, which concludes the proof.
\end{proof}

\begin{lemma*} 
    Protocol~\ref{ERSprot} is $\varepsilon_{\rm rep}$-secure against repudiation (or non-transferability) according to Definition~\ref{def:transferability-ERS}, with $\abs{C}<N/3$ and
    \begin{align}
        \varepsilon_{\rm rep}= (b_M+b_H + y)/2^{b'_H-1}.
    \end{align}
\end{lemma*}

\begin{proof}
    In the first part of this proof, we assume ideal authenticated channels between all receivers pairs. Then, by virtue of the security properties of the BRB protocol (see Appendix~\ref{app:BRB}) and under the assumption that the dishonest receivers are strictly less than $N/3$, we are guaranteed that all honest receivers end up with the same $\{Doc,Sig\}$ pair and the same string $X_r$. In other words, we are certain that:
    \begin{align}
        Doc_i &= Doc \quad \forall\, P_i \notin C \\
        Sig_i &= Sig \quad \forall\, P_i \notin C \\
        X_{r,i} &= X_r \quad \forall\, P_i \notin C.
    \end{align}
    Therefore, all honest receivers must agree on the validity (or rejection) of the signature. Alternatively, if the BRB protocol aborts in step 2.3 or 2.5, it does so for all honest receivers. In conclusion, we have that:
    \begin{align}
        \mathrm{Ver}_{i}(Doc_i,Sig_i) = \mathrm{Ver}_{j}(Doc_j,Sig_j) \quad \forall\, P_i,P_j \notin C,
    \end{align}
    which implies that the ERS protocol is $\varepsilon_{\rm rep}=0$ secure against repudiation.

    Now, we relax the assumption of ideal authenticated channels and consider their implementation with hash functions from the $\mathcal{F}_{\rm AXU}$ family, as per Sec.~\ref{sec:notation}. In this case, honest receivers could disagree on the outcome of signature verification if the coalition $C$ of dishonest parties successfully attacks the BRB protocol in step 2.3 or step 2.5 of Protocol~\ref{ERSprot}. Indeed, when attacking the authenticated channels in the first BRB protocol, the coalition could enforce that an honest receiver terminates the BRB protocol with a $\{Doc_i,Sig_i\}$ pair different from the others. Similarly, attacking the second BRB protocol could cause an honest receiver to obtain a string $X_{r,i}$ different from the strings of the other honest parties. Alternatively, attacking the authenticated channels could also cause an asymmetric abort among honest receivers.
    
    To obtain an upper bound on the probability of a successful repudiation attack, we consider the success probability of an attack on one authenticated message exchanged in the BRB protocols of step 2.3 and step 2.5, given by $b_{\rm text}/2^{b'_H-1}$ where $b_{\rm text}$ is the bit-length of the message. Since the messages exchanged in the first BRB protocol are the $\{Doc,Sig\}$ pairs, we have $b_{\rm text}=b_M+b_H + y$. While in the second BRB protocol, we have: $b_{\rm text}=y$. Since we assumed that the ERS protocol aborts for everyone upon the failure of message authentication, the maximum probability of a successful repudiation attack is given by:
    \begin{align}
        \max \frac{\left\lbrace b_M+b_H + y , y \right\rbrace }{2^{b'_H-1}},
    \end{align}
    which yields the claim. This concludes the proof.
\end{proof}

\end{document}